\documentclass[aps,notitlepage,twocolumn,nofootinbib,superscriptaddress,longbibliography]{revtex4-2}
\usepackage{ulem}
\usepackage{algorithmic}
\usepackage[ruled, vlined, linesnumbered]{algorithm2e}

\usepackage{amsmath,amsfonts,amssymb}
\usepackage{mathtools}
\usepackage{multirow}
\usepackage[utf8]{inputenc}
\usepackage[T1]{fontenc}

\usepackage{bm}
\usepackage{times}

\usepackage{mathtools}
\usepackage{amsmath}
\usepackage[shortlabels]{enumitem}

\usepackage{graphicx,epic,eepic,epsfig,amsmath,latexsym,amssymb,verbatim,color}
 
\usepackage{amsfonts}       
\usepackage{nicefrac}       

\usepackage{amsmath}
\usepackage{bbm}

\usepackage{float}
\usepackage{tikz}
\usetikzlibrary{chains}
\usetikzlibrary{fit}
\usetikzlibrary{arrows}

\usepackage{epsfig}
\usetikzlibrary{shapes.symbols,patterns} 
\usepackage{pgfplots}
\pgfplotsset{compat=1.18}

\usepackage[strict]{changepage}
\usepackage{hyperref}
\hypersetup{colorlinks=true,citecolor=blue,linkcolor=blue,filecolor=blue,urlcolor=blue,breaklinks=true}

\usepackage[marginal]{footmisc}
\usepackage{url}
\usepackage{theorem}

\newtheorem{proposition}{Proposition}
\newtheorem{lemma}[proposition]{Lemma}

\newtheorem{theorem}[proposition]{Theorem}

\def\squareforqed{\hbox{\rlap{$\sqcap$}$\sqcup$}}
\def\qed{\ifmmode\squareforqed\else{\unskip\nobreak\hfil
\penalty50\hskip1em\null\nobreak\hfil\squareforqed
\parfillskip=0pt\finalhyphendemerits=0\endgraf}\fi}
\def\endenv{\ifmmode\;\else{\unskip\nobreak\hfil
\penalty50\hskip1em\null\nobreak\hfil\;
\parfillskip=0pt\finalhyphendemerits=0\endgraf}\fi}
\newenvironment{proof}{\noindent \textbf{{Proof~} }}{\hfill $\blacksquare$}

\newcounter{remark}

\newcounter{example}

\mathchardef\ordinarycolon\mathcode`\:
\mathcode`\:=\string"8000
\def\vcentcolon{\mathrel{\mathop\ordinarycolon}}
\begingroup \catcode`\:=\active
  \lowercase{\endgroup
  \let :\vcentcolon
  }

\usepackage{cleveref}
\usepackage{graphicx}
\usepackage{xcolor}

\RequirePackage[framemethod=default]{mdframed}
\newmdenv[skipabove=7pt,
skipbelow=7pt,
backgroundcolor=darkblue!15,
innerleftmargin=5pt,
innerrightmargin=5pt,
innertopmargin=5pt,
leftmargin=0cm,
rightmargin=0cm,
innerbottommargin=5pt,
linewidth=1pt]{tBox}

\newmdenv[skipabove=7pt,
skipbelow=7pt,
backgroundcolor=blue2!25,
innerleftmargin=5pt,
innerrightmargin=5pt,
innertopmargin=5pt,
leftmargin=0cm,
rightmargin=0cm,
innerbottommargin=5pt,
linewidth=1pt]{dBox}
\newmdenv[skipabove=7pt,
skipbelow=7pt,
backgroundcolor=darkkblue!15,
innerleftmargin=5pt,
innerrightmargin=5pt,
innertopmargin=5pt,
leftmargin=0cm,
rightmargin=0cm,
innerbottommargin=5pt,
linewidth=1pt]{sBox}
\definecolor{darkblue}{RGB}{0,76,156}
\definecolor{darkkblue}{RGB}{0,0,153}
\definecolor{blue2}{RGB}{102,178,255}
\definecolor{darkred}{RGB}{195,0,0}

\newcommand{\nc}{\newcommand}
\nc{\rnc}{\renewcommand}
\nc{\lbar}[1]{\overline{#1}}
\nc{\bra}[1]{\langle#1|}
\nc{\ket}[1]{|#1\rangle}
\nc{\ketbra}[2]{|#1\rangle\!\langle#2|}
\nc{\braket}[2]{\langle#1|#2\rangle}

\nc{\proj}[1]{| #1\rangle\!\langle #1 |}
\nc{\avg}[1]{\langle#1\rangle}
\nc{\rank}{\operatorname{Rank}}
\nc{\smfrac}[2]{\mbox{$\frac{#1}{#2}$}}
\nc{\tr}{\operatorname{Tr}}
\nc{\ox}{\otimes}
\nc{\dg}{\dagger}
\nc{\dn}{\downarrow}
\nc{\cA}{{\cal A}}
\nc{\cB}{{\cal B}}
\nc{\cC}{{\cal C}}
\nc{\cD}{{\cal D}}
\nc{\cE}{{\cal E}}
\nc{\cF}{{\cal F}}
\nc{\cG}{{\cal G}}
\nc{\cH}{{\cal H}}
\nc{\cI}{{\cal I}}
\nc{\cJ}{{\cal J}}
\nc{\cK}{{\cal K}}
\nc{\cL}{{\cal L}}
\nc{\cM}{{\cal M}}
\nc{\cN}{{\cal N}}
\nc{\cO}{{\cal O}}
\nc{\cP}{{\cal P}}
\nc{\cQ}{{\cal Q}}
\nc{\cR}{{\cal R}}
\nc{\cS}{{\cal S}}
\nc{\cT}{{\cal T}}
\nc{\cU}{{\cal U}}
\nc{\cV}{{\cal V}}
\nc{\cX}{{\cal X}}
\nc{\cY}{{\cal Y}}
\nc{\cZ}{{\cal Z}}
\nc{\cW}{{\cal W}}

\newcommand{\bbE}{\mathbb{E}}
\newcommand{\bbP}{\mathbb{P}}

\nc{\csupp}{{\operatorname{csupp}}}
\nc{\qsupp}{{\operatorname{qsupp}}}
\nc{\var}{{\operatorname{var}}}
\nc{\rar}{\rightarrow}
\nc{\lrar}{\longrightarrow}
\nc{\polylog}{{\operatorname{polylog}}}
\nc{\wt}{{\operatorname{wt}}}
\nc{\av}[1]{{\left\langle {#1} \right\rangle}}
\nc{\supp}{{\operatorname{supp}}}

\nc{\argmin}{{\operatorname{argmin}}}

\def\x{\xi}

\nc{\RR}{{{\mathbb R}}}
\nc{\CC}{{{\mathbb C}}}
\nc{\FF}{{{\mathbb F}}}
\nc{\NN}{{{\mathbb N}}}
\nc{\ZZ}{{{\mathbb Z}}}
\nc{\PP}{{{\mathbb P}}}
\nc{\QQ}{{{\mathbb Q}}}
\nc{\UU}{{{\mathbb U}}}
\nc{\EE}{{{\mathbb E}}}
\nc{\id}{{\operatorname{id}}}

\nc{\CHSH}{{\operatorname{CHSH}}}

\nc{\be}{\begin{equation}}
\nc{\ee}{{\end{equation}}}
\nc{\bea}{\begin{eqnarray}}
\nc{\eea}{\end{eqnarray}}
\nc{\<}{\langle}
\rnc{\>}{\rangle}
\nc{\rU}{\mbox{U}}

\nc{\ob}[1]{#1}

\nc{\SEP}{{\text{\rm SEP}}}
\nc{\NS}{{\text{\rm NS}}}
\nc{\LOCC}{{\text{\rm LOCC}}}
\nc{\PPT}{{\text{\rm PPT}}}
\nc{\EXT}{{\text{\rm EXT}}}
\nc{\Sym}{{\operatorname{Sym}}}

\nc{\ERLO}{{E_{\text{r,LO}}}}
\nc{\ERLOCC}{{E_{\text{r,LOCC}}}}
\nc{\ERPPT}{{E_{\text{r,PPT}}}}
\nc{\ERLOCCinfty}{{E^{\infty}_{\text{r,LOCC}}}}
\nc{\Aram}{{\operatorname{\sf A}}}

\usepackage{tikz}
\usepackage{hyperref}
\hypersetup{colorlinks=true,citecolor=blue,linkcolor=blue,filecolor=blue,urlcolor=blue,breaklinks=true}

\makeatletter
\def\grd@save@target#1{%
  \def\grd@target{#1}}
\def\grd@save@start#1{%
  \def\grd@start{#1}}
\tikzset{
  grid with coordinates/.style={
    to path={%
      \pgfextra{%
        \edef\grd@@target{(\tikztotarget)}%
        \tikz@scan@one@point\grd@save@target\grd@@target\relax
        \edef\grd@@start{(\tikztostart)}%
        \tikz@scan@one@point\grd@save@start\grd@@start\relax
        \draw[minor help lines,magenta] (\tikztostart) grid (\tikztotarget);
        \draw[major help lines] (\tikztostart) grid (\tikztotarget);
        \grd@start
        \pgfmathsetmacro{\grd@xa}{\the\pgf@x/1cm}
        \pgfmathsetmacro{\grd@ya}{\the\pgf@y/1cm}
        \grd@target
        \pgfmathsetmacro{\grd@xb}{\the\pgf@x/1cm}
        \pgfmathsetmacro{\grd@yb}{\the\pgf@y/1cm}
        \pgfmathsetmacro{\grd@xc}{\grd@xa + \pgfkeysvalueof{/tikz/grid with coordinates/major step}}
        \pgfmathsetmacro{\grd@yc}{\grd@ya + \pgfkeysvalueof{/tikz/grid with coordinates/major step}}
        \foreach \x in {\grd@xa,\grd@xc,...,\grd@xb}
        \node[anchor=north] at (\x,\grd@ya) {\pgfmathprintnumber{\x}};
        \foreach \y in {\grd@ya,\grd@yc,...,\grd@yb}
        \node[anchor=east] at (\grd@xa,\y) {\pgfmathprintnumber{\y}};
      }
    }
  },
  minor help lines/.style={
    help lines,
    step=\pgfkeysvalueof{/tikz/grid with coordinates/minor step}
  },
  major help lines/.style={
    help lines,
    line width=\pgfkeysvalueof{/tikz/grid with coordinates/major line width},
    step=\pgfkeysvalueof{/tikz/grid with coordinates/major step}
  },
  grid with coordinates/.cd,
  minor step/.initial=.2,
  major step/.initial=1,
  major line width/.initial=2pt,
}
\makeatother

\usepackage{thmtools}
\usepackage{thm-restate}
\usepackage{etoolbox}
\makeatletter
\def\problem@s{}
\newcounter{problems@cnt}

\newcommand{\allproblems}{\problem@s}
\makeatother

\definecolor{beamer}{rgb}{0.2,0.2,0.7}
\definecolor{colorone}{rgb}{1,0.36,0.03}
\definecolor{colortwo}{rgb}{0.4,0.77,0.17}
\definecolor{colorthree}{rgb}{0.01,0.51,0.93}
\definecolor{colorfour}{rgb}{0.47,0.26,0.58}
\definecolor{colorfive}{rgb}{0.12,0.55,0.16}
\usepackage{tcolorbox}
\usepackage{relsize}
\usepackage{graphicx}
\usepackage{booktabs}
\usepackage{amsmath}
\usepackage{tikz}
\usetikzlibrary{quantikz}
\nc{\st}{\text{subject to} \ }
\nc{\supre}{\text{supremum} \ }
\nc{\sdp}{\text{sdp}}
\usepackage{array}

\newcommand{\veps}{\varepsilon}

\newcommand{\vtheta}{\vartheta}

\newcommand{\ri}{i} 

\newcommand{\hS}{\hat{S}}

\newcommand{\stat}{\mathrm{stat}}
\newcommand{\sys}{\mathrm{sys}}

\newcommand{\GHZ}{\mathrm{GHZ}}

\newcommand{\hw}{\mathrm{hw}} 
\newcommand{\bias}{\mathrm{bias}}

\newcommand{\htheta}{\hat{\theta}}
\newcommand{\fail}{\mathrm{fail}}
\newcommand{\std}{\mathrm{std}} 

\newcommand{\bigO}{\mathcal{O}}
\newcommand{\tar}{\mathrm{tar}}

\allowdisplaybreaks

\nc{\ith}[1]{{#1}^\mathrm{th}}

\begin{document}
\title{Efficient quantum phase estimation with adaptive entanglement-assisted Hadamard test}

\author{Hengzhun Chen}
\affiliation{School of Mathematical Sciences, Fudan University}

\author{Benchi Zhao}
\email{benchizhao@gmail.com}
\affiliation{QICI Quantum Information and Computation Initiative, School of Computing and Data Science, The University of Hong Kong, Pokfulam Road, Hong Kong}

\author{Yingzhou Li}
\email{yingzhouli@fudan.edu.cn}
\affiliation{School of Mathematical Sciences, Fudan University}
\affiliation{Shanghai Key Laboratory for Contemporary Applied Mathematics}

\begin{abstract}
The entanglement-assisted Hadamard test (EHT) is a practical method for estimating a quantum phase by amplifying the phase signal. However, the feasible amplification is fundamentally limited by the accuracy of the reference phase, such that the method is inefficient in the high-precision regime. In this work, we propose an algorithm, called \textit{adaptive entanglement-assisted Hadamard test} (AEHT), that iteratively refines the reference phase, enabling progressively stronger amplification as the iteration goes by. We further consider the imperfect eigenstate preparation scenario, where a systematic bias is unavoidable when estimating the quantum phase with the conventional EHT. Such a bias can be suppressed by the proposed AEHT. Moreover, taking physical implementation into consideration, we adopt \textit{device-restart count} to measure the cost of quantum phase estimation, rather than shot count. The numerical experiments confirm the effectiveness of the proposed AEHT compared with conventional methods under this measure. By unlocking the full amplification power of large entangled states, this work offers an efficient method to estimate high-precision quantum phase on near-term quantum processors.
\end{abstract}

\maketitle

\section{Introduction}
Quantum computation is expected to offer exponential advantages over its classical counterparts, potentially enabling the efficient solution of problems that are intractable for classical computers. Quantum phase estimation (QPE) is one of the most fundamental subroutines in quantum computing, providing a systematic procedure for extracting the eigenphase of a unitary operator given access to the corresponding eigenstate~\cite{kitaev1995quantum,nielsen2010quantum,dorner2009optimal, abrams1999quantum}. QPE constitutes the computational backbone of Shor's factoring algorithm~\cite{shor1999polynomial, lanyon2007experimental, martin2012experimental,monz2016realization}, the Harrow–Hassidim–Lloyd algorithm for linear systems of equations~\cite{harrow2009quantum, cai2013experimental}, quantum entropy estimation~\cite{wang2023quantum}, and quantum simulation of molecular Hamiltonians~\cite{cruz2020optimizing, santagati2018witnessing, lanyon2010towards, du2010nmr, o2016scalable, PhysRevA.110.L060401}. However, estimating the quantum phase by conventional QPE~\cite{kitaev1995quantum, nielsen2010quantum} demands coherence times that scale exponentially with the desired precision and are highly sensitive to experimental noise~\cite{wiebe2016efficient, rubin2007loss}. Consequently, without full quantum error correction, conventional QPE algorithms quickly become impractical, which severely limits their applicability on near-term quantum devices.

Standard Hadamard test (SHT)~\cite{dobvsivcek2007arbitrary, nielsen2010quantum} provides a practical alternative with shallow circuits to estimate the quantum phase, which extracts the quantum phase $\theta$ through repeated measurements (as shown in FIG.~\ref{fig:hadamard-test}). Specifically, suppose a coarse reference phase $\vtheta$ is known, then by conducting a quantum circuit for multiple shot counts, the expectation value $\hat{S}$ corresponds to the estimate of $\sin(\theta-\vtheta)$, and therefore the estimated phase can be directly calculated $\hat{\theta}=\arcsin(\hat{S})+\vtheta$. Its statistical error nevertheless decreases only as the inverse square root of the shot counts, so estimating a phase to accuracy $\epsilon$ generally requires $\cO(1/\epsilon^{2})$ shot counts. 

For a further step, entanglement-assisted Hadamard test (EHT)~\cite{giovannetti2011advances, giovannetti2004quantum, campos2003optical, nielsen2023deterministic} use $m$-qubit Greenberger-Horne-Zeilinger (GHZ) entangled state and parallel controlled unitary to amplify phase signal, i.e., the expectation value $\hat{S}$ corresponds to the estimate of $\sin(m(\theta-\vtheta))$, and the quantum phase is estimated as $\hat{\theta}=\frac{1}{m}\arcsin(\hat{S})+\vtheta$ (as shown in FIG.~\ref{fig:ghz-hadamard-general}). With the assistance of quantum entanglement, to achieve an estimated phase with accuracy $\epsilon$, the required number of shot counts reduces to $\cO(1/(m\epsilon)^2)$. The larger the GHZ state that is utilized, the fewer shot counts are required. But in practice, the GHZ state size $m$ cannot be infinitely large, as large $m$ leads to $m|\theta-\vtheta|> \frac{\pi}{2}$, making the $\arcsin (\hat{S})$ ambiguous. In other words, a coarse reference phase therefore permits only a small amplification, preventing the direct use of large GHZ states even when sufficient hardware resources are available. This limitation becomes particularly important when high precision is required. The ambiguity problem can be avoided by estimating both $\sin(\theta-\vtheta)$ and $\cos(\theta-\vtheta)$. But meanwhile, the shot counts are in doubt, making it inefficient.

In this work, we address this limitation by developing an algorithm, called \textit{adaptive entanglement-assisted  Hadamard test} (AEHT).  Starting from a coarse reference phase, we first use a modestly sized GHZ state whose amplification is safe given the current uncertainty. The measurements sharpen the estimate, and this improved estimate becomes the reference for the next round. With a tighter reference, a larger GHZ state becomes admissible, so the next round can amplify the remaining phase signal more strongly. Repeating this process allows the amplification to grow as the reference improves, keeping the required shot counts under control even at high precision.
We also extend this procedure to the more realistic situation in which the prepared state is not an exact eigenstate. There the Hadamard signal has the form $\rho e^{i\theta_\psi}$, where the amplitude $\rho$ is smaller than one. In reality, the amplitude is generally unknown, inducing a bias in phase estimation that cannot be removed even if more shot counts are provided. We show that the proposed AEHT algorithm suppresses this bias, such that the estimation error continues to decrease round by round.

Adaptive phase estimation has been explored before, notably in entanglement-free interferometric schemes that use repeated phase shifts and Bayesian feedback~\cite{higgins2007entanglement}. Those methods, however, rely on many sequential applications of the unknown phase and therefore require long coherence times that remain unrealistic on near-term hardware. Here we target a more practical regime, in which the coherence time is limited but the processor instead provides many qubits that can be entangled and operated in parallel. We therefore keep the entanglement-assisted Hadamard test and make the reference adaptive, using GHZ amplification to trade for fewer shot counts.

The shot count has been used to quantify the realization cost of quantum algorithms~\cite{dorner2009optimal,ni2023low,wang2023quantum}. However, such a measure is unfair when comparing EHT with different sizes of $m$-qubit GHZ entangled states. The shot with large GHZ entangled states takes up more quantum systems, while the shot with small GHZ entangled states takes up fewer quantum systems. The available quantum systems is also a critical resource in quantum computing, but the shot count does not take this into account. In fact, on a large quantum processor, many shots can be executed in parallel within a single experimental cycle. For a quantum device, the truly time-consuming operations are loading, rearranging, and initializing, which is particularly pronounced on neutral-atom platforms~\cite{bluvstein2024logical, henriet2020quantum,saffman2016quantum, ebadi2021quantum}. Therefore, we adopt the \textit{device-restart count} to quantify the realization cost, instead of the shot count. Under this measure, the advantage of the adaptive scheme becomes clear. A fixed-amplification EHT must repeat the same circuit until the required precision is reached, so its restart count grows as $\mathcal{O}(1/\epsilon^2)$. In contrast, the AEHT progressively enlarges the GHZ state as the reference improves, keeping the per-round shot budget roughly constant; its restart count therefore grows only as $\mathcal{O}(1/\epsilon)$, up to logarithmic factors, and the ratio to the fixed-amplification cost vanishes as $\epsilon\to0$. When the GHZ size is eventually limited by the available number of qubits, the adaptive scheme saturates and the scaling returns to $\mathcal{O}(1/\epsilon^2)$, but it still retains a constant-factor advantage over a fixed-amplification method. The numerical experiments are conducted to demonstrate the advantage of the proposed AEHT over the conventional EHT in terms of device-restart count. This work makes entanglement-assisted phase estimation practical in the high-precision regime by adapting the reference to the estimate and by measuring cost in terms of device-restart counts rather than shot counts.

The paper is organized as follows. In Sec.~\ref{sec:Hadamard-test-based Phase Estimation} we review the standard and entanglement-assisted Hadamard tests and their shot count requirements. Section~\ref{sec:Adaptive entanglement-assisted Hadamard test} introduces the adaptive entanglement-assisted Hadamard test for exact eigenstates and analyzes how the reference can be refined to permit larger amplification. Section~\ref{sec:imperfect} extends the AEHT to imperfect state preparation scenario. Section~\ref{sec:comarison} introduces the device-restart count metric, compares fixed-amplification and adaptive schemes, and discusses the advantage as well as the effect of hardware-limited amplification. Section~\ref{sec:numerical_experiments} presents numerical experiments for exact and imperfect eigenstates, including a circuit-noise robustness study. We conclude and discuss open directions in Sec.~\ref{sec:conclusion}.


\section{Hadamard-test-based Phase Estimation}\label{sec:Hadamard-test-based Phase Estimation}
Estimating the quantum phase associated with the eigenvalue of a unitary operator $U$ is a fundamental task in quantum algorithms~\cite{kitaev1995quantum, nielsen2010quantum}. Given a unitary $U$ and one of its eigenstates $\ket{\phi}$ satisfying $U\ket{\phi} = e^{i\theta}\ket{\phi}$, the goal of phase estimation is to estimate the eigenphase $\theta$ to some desired precision. Conventional QPE algorithm, as demonstrated in Appendix~\ref{app:quantum_phase_estimation}, requires deep quantum circuits, which in turn demand long coherence times from the quantum hardware. However, near-term quantum hardware cannot meet this requirement.
Here we introduce two commonly used alternatives to estimate quantum phase: the standard Hadamard test and the entanglement-assisted Hadamard test. 
We also establish explicit finite-sample shot-count bounds that account for reference-phase uncertainty. These bounds quantify how a sufficient shot count depends jointly on the reference-phase uncertainty, amplification, target accuracy, and failure probability.

\subsection{Phase estimation with standard Hadamard test}\label{sec:phase_estimation_with_HT}
We first introduce the standard Hadamard test to estimate the eigenphase $\theta$. The quantum circuit for phase estimation with the standard Hadamard test is shown in Fig.~\ref{fig:hadamard-test}. At the beginning, we give a reference phase $\vtheta$, which is a guess of the true phase. Then, we prepare one ancilla in $\ket{0}$, apply a Hadamard gate, and controlled-$U$ to one copy of the eigenstate $\ket{\phi}$ of unitary $U$. Via phase kickback, the ancilla state becomes $\frac{1}{\sqrt{2}}(\ket{0} + e^{i\theta}\ket{1})$. Then, apply the single-qubit reference phase gate $\Phi(\vtheta)$, which maps $\ket{0}\rightarrow\ket{0}$ and $\ket{1}\rightarrow-\ri e^{-\ri \vtheta}\ket{1}$. Apply the final Hadamard gate and measure the ancilla. The measurement is denoted as $X\in\{-1,1\}$. With $N$ independent shots, the average value is $\hS:=\frac{1}{N}\sum_{j=1}^{N}X^{(j)}$. 

\begin{figure}[h]
    \centering
    \includegraphics[width=0.8\linewidth]{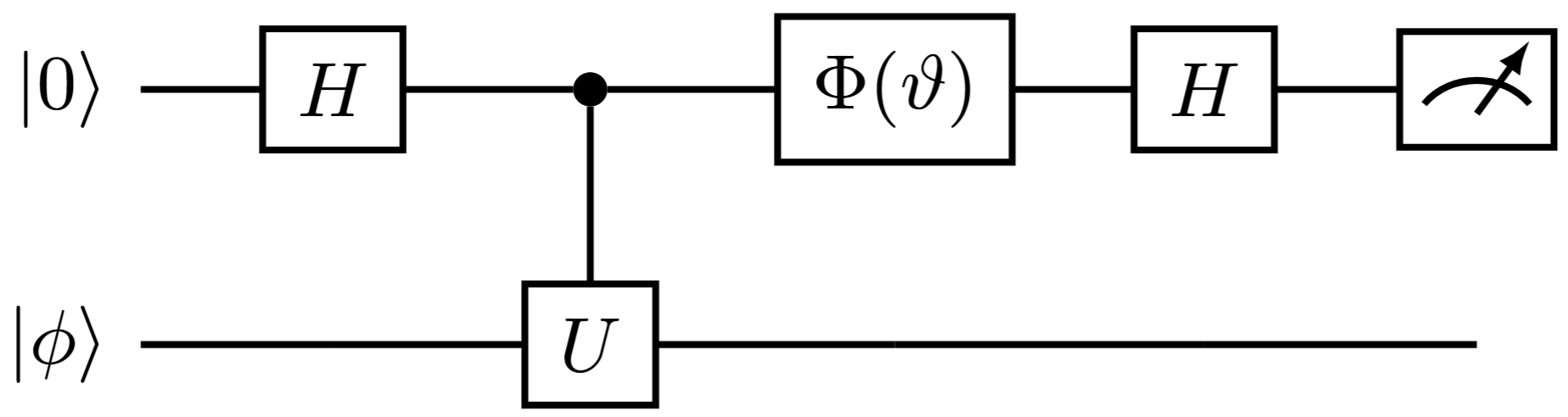}
    \caption{Quantum circuit of standard Hadamard test with reference.}
    \label{fig:hadamard-test}
\end{figure}

The expectation value of the measurement results is $\mathbb{E}[\hat{S}] = \sin(\theta-\vtheta)$, which is derived in Appendix~\ref{app:hadamard_expectation_values}.  Therefore $\hat{S}$ estimates $\sin(\theta-\vtheta)$. The quantum phase can be directly estimated by
\begin{equation}\label{eq:estimate_phase_std}
    \htheta
    =
    \vtheta+\arcsin (\hS).
\end{equation}
Note that $\arcsin (\hS)$ might not be unique. To avoid the ambiguity, one has to guarantee the reference phase is close to the true phase, i.e., $|\theta - \vtheta|<\pi/2$. In reality, the value $|\theta-\vtheta|$ and the true phase $\theta$ are unknown. But one can conduct a few measurements to get a coarse estimation, and then set the estimated coarse phase as the reference. In this case, one can assume $|\theta-\vtheta|\le \Delta$, where $\Delta<\pi/2$ is some constant.

According to Eq.~\eqref{eq:estimate_phase_std}, the phase error propagated from the sine-estimation error is approximately
$$\delta_{\hat{\theta}} \approx \frac{\delta_{\hat{S}}}{\sqrt{1-\hat{S}^2}}\lesssim\frac{\delta_{\hat{S}}}{\sqrt{1-\sin^2(\Delta)}},$$
where $(1-\hat{S}^2)^{-1/2}$ is the derivative of $\arcsin(\hat{S})$. If $\Delta$ approaches $\pm \pi/2$, the upper bound of the estimation error could go to infinity. In this case, to achieve the required phase estimation accuracy $\delta_{\hat{\theta}} \le \epsilon$, the estimation error $\delta_{\hat{S}}$ is supposed to be small enough, and the shot count is supposed to be sufficiently large, which is unacceptable in practice. Therefore, to guarantee the efficiency of phase estimation, the upper bound $\Delta$ is supposed to be sufficiently small. The analysis of the shot count is given in Lemma~\ref{lem:standard_hadamard_shot_count}.

\begin{lemma}[Shot count for the SHT]
    \label{lem:standard_hadamard_shot_count}
    Let $\ket{\phi}$ be an eigenstate of a unitary $U$ with $U\ket{\phi}=e^{\ri\theta}\ket{\phi}$, and let $\vtheta$ be a reference phase satisfying $|\theta-\vtheta|\le\Delta$ for some constant $\Delta<\pi/2$. To estimate $\theta$ via the standard Hadamard test within additive error $\epsilon$ and with success probability at least $1-p_{\fail}$, it suffices to take a shot count
    \begin{equation*}
        N_{\std}
        =
        \left\lceil\frac{2}{r_{\std}^2}\,\ln\frac{2}{p_{\fail}}\right\rceil,
    \end{equation*}
    where $r_{\rm std} = \epsilon\left(\sqrt{\cos^2\Delta+\epsilon^2}
            -\epsilon\sin\Delta\right)/(1+\epsilon^2)$.
\end{lemma}

The proof of Lemma~\ref{lem:standard_hadamard_shot_count} is shown in Appendix~\ref{app:proof_of_std_hadamard}. The value $r_{\rm std}$ can be written as the Taylor expansion form, i.e., $r_{\rm std} = \epsilon\cos\Delta + \cO(\epsilon^2)$. Since the target accuracy is generally tiny, the higher-order terms $\cO(\epsilon^2)$ are negligible. Thus, the shot count complexity of phase estimation with standard Hadamard test is $\cO(1/\epsilon^2)$ for target accuracy $\epsilon$.

\subsection{Phase estimation with entanglement-assisted Hadamard test}\label{sec:Phase estimation with entanglement-assisted Hadamard test}

The standard Hadamard test requires $\cO(1/\epsilon^2)$ shots, which is expensive on near-term devices. One way to reduce this shot count is to use quantum entanglement~\cite{giovannetti2011advances, giovannetti2004quantum, campos2003optical, nielsen2023deterministic}.  Specifically, prepare an $m$-qubit Greenberger-Horne-Zeilinger (GHZ) state as ancilla system
\begin{equation*}
    \ket{\GHZ_m}
    :=\frac{1}{\sqrt{2}}\left(\ket{0}^{\otimes m}+\ket{1}^{\otimes m}\right).
\end{equation*}
Then apply controlled-$U$ in parallel to the $m$-copies of the eigenstate $\ket{\phi}$ of unitary $U$. Afterwards, the phase block $\Phi_m(\vtheta)$ with reference phase $\vtheta$ acts on the GHZ ancilla systems, which maps $\ket{0}^{\otimes m}\rightarrow\ket{0}^{\otimes m}$ and $\ket{1}^{\otimes m}\rightarrow-\ri e^{-\ri m\vtheta}\ket{1}^{\otimes m}$. After the final Hadamards, measure the ancilla qubits in the computational basis and denote the outcome as $z\in\{0,1\}^m$, and the Hamming weight of the measurements is defined as $|z|:=\sum_{j=1}^m z_j$. The quantum circuit for the phase estimation with entanglement-assisted Hadamard test is shown in Fig.~\ref{fig:ghz-hadamard-general}.

\begin{figure}[h]
    \centering
    \includegraphics[width=\linewidth]{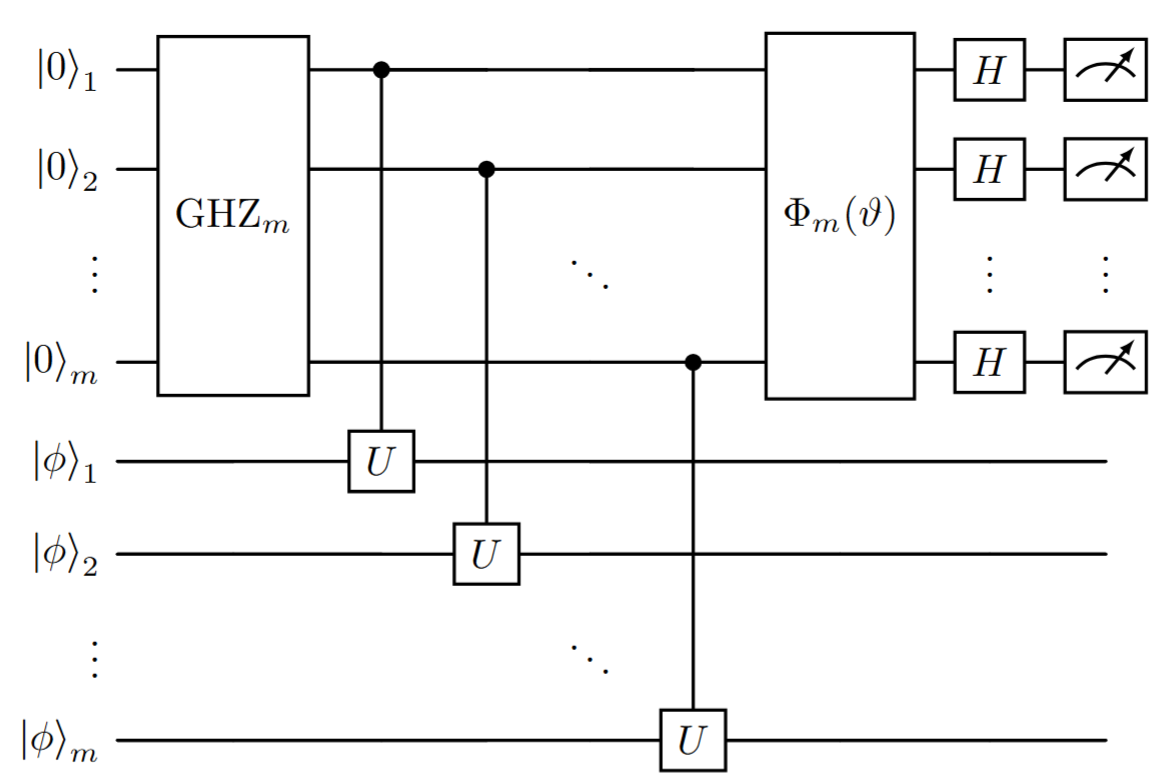}
    \caption{Quantum circuit of entanglement-assisted Hadamard test with $m$-qubit GHZ state.}
    \label{fig:ghz-hadamard-general}
\end{figure}

Denote $Y := (-1)^{|z|}\in\{+1,-1\}$, and the corresponding average value over $N$ independent shots is $\hat{S} = \frac{1}{N}\sum_{j=1}^N Y^{(j)}$. The expectation value of the measurement results is $\mathbb{E}[Y] = \sin(m(\theta-\vtheta))$, which is derived in Appendix~\ref{app:hadamard_expectation_values}. Therefore, $\hat{S}$ estimates $\sin(m(\theta-\vtheta))$. Note that the phase signal is amplified by $m$ times, thus the quantity $m$ is called the \textit{amplification}. The quantum phase can be directly estimated by
\begin{equation}\label{eq:phase_estimation_ent}
    \hat{\theta}:=\vtheta+\frac{1}{m}\arcsin (\hS).
\end{equation}
Due to the ambiguity of $\arcsin (\hat S)$, similar to the standard Hadamard test, we require the residual phase to lie on an unambiguous branch, i.e., $|\theta-\vartheta|\le\Delta$ with $m\Delta<\pi/2$.

The phase error propagated from the sine-estimation error is approximately
\begin{equation}\label{eq:paraller_error}
    \delta_{\hat{\theta}} \approx \frac{\delta_{\hat{S}}}{m\sqrt{1-\hat{S}^2}}\lesssim \frac{\delta_{\hat{S}}}{m\sqrt{1-\sin^2(m\Delta)}},
\end{equation}
where $(1-\hat{S}^2)^{-1/2}$ is the derivative of $\arcsin(\hat{S})$. When $m\Delta$ approaches to $\pm \pi/2$, the upper bound of the estimation error could go to infinity. To address this problem, the choice of amplification $m$ should be constrained by $m\le \lfloor\chi/\Delta\rfloor$, where $\chi<\pi/2$ is a reasonable constant called \textit{branch margin}. The shot count to obtain the phase within the required accuracy is given in Lemma~\ref{lem:parallel_hadamard_shot_count}.

\begin{lemma}[Shot count for the EHT]
    \label{lem:parallel_hadamard_shot_count}
    Let $\ket{\phi}$ be an eigenstate of a unitary $U$ with $U\ket{\phi}=e^{\ri\theta}\ket{\phi}$, and let $\vtheta$ be a reference phase satisfying $|\theta-\vtheta|\le \Delta$ for some constant $\Delta<\frac{\pi}{2m}$. Fix a branch margin $\chi<\pi/2$ and choose an integer amplification $m$ such that $m\Delta\le\chi$. To estimate the phase $\theta$ via the $m$-qubit GHZ state EHT within additive error $\epsilon$ and with success probability at least $1-p_{\fail}$, it suffices to take a shot count
    \begin{equation*}
        N_{\rm EHT}=\left\lceil\frac{2}{r_{\rm EHT}^2}\ln\frac{2}{p_{\fail}}\right\rceil,
    \end{equation*}
    where 
    \begin{align}
        r_{\rm EHT} = \frac{m\epsilon\left(\sqrt{\cos^2(m\Delta)+m^2\epsilon^2}-m\epsilon\sin(m\Delta)\right)}{1+m^2\epsilon^2}.
    \end{align}
\end{lemma}

The proof is shown in Appendix~\ref{app:proof_of_parallel_hadamard_shot_count}. The value $r_{\rm EHT}$ can be written as the Taylor expansion form $r_{\rm EHT} = m\epsilon\cos(m\Delta) + \bigO((m\epsilon)^2)$. Since the target accuracy $\epsilon$ is generally tiny, the high-order term $\bigO((m\epsilon)^2)$ is negligible. Therefore, the shot count of phase estimation with $m$-qubit GHZ state EHT is 
\begin{equation}\label{eq:N_EHT}
    N_{\rm EHT} \approx
    \frac{2}{m^2\epsilon^2\cos^2(m\Delta)}
    \ln\frac{2}{p_{\mathrm{fail}}}.
\end{equation}
Considering the amplification that is upper bounded by $m\le \lfloor\chi/\Delta\rfloor$, the natural choice of amplification is $m= \lfloor\chi/\Delta\rfloor$, because in the general case, the larger amplification $m$ implies fewer shot counts. Thus, we roughly have $m\Delta\approx \chi$, and the shot count complexity becomes $N_{\rm EHT}=\cO(1/(m\epsilon)^2)$. Compared with the SHT, the EHT reduces the required number of shots by a factor of $m^2$. If we ignore the constant $m$, then the shot count scales $N_{\rm EHT}=\cO(1/\epsilon^2)$.
The algorithm of phase estimation with EHT in pseudo-code is provided in Algorithm~\ref{alg:single_round_parallel_hadamard}. By setting the amplification $m=1$, the algorithm collapses to the standard Hadamard test as introduced in Sec.~\ref{sec:phase_estimation_with_HT}.

\begin{algorithm}[htbp]
\renewcommand{\algorithmicrequire}{\textbf{Input:}}
\renewcommand{\algorithmicensure}{\textbf{Output:}}
    \caption{Phase estimation with EHT}
    \label{alg:single_round_parallel_hadamard}
    \begin{algorithmic}[1]
        \REQUIRE Reference $\vtheta$, amplification $m$, reference bound $\Delta$, accuracy $\epsilon$, failure probability $p_{\fail}$.
        \ENSURE The estimated phase $\hat \theta$.
        \STATE $r\gets
            \frac{\left(\sqrt{\cos^2(m\Delta)+m^2\epsilon^2}-m\epsilon\sin(m\Delta)\right)m\epsilon}{1+m^2\epsilon^2} $
        \STATE $N\gets
            \left\lceil\dfrac{2}{r^2}
            \ln\dfrac{2}{p_{\fail}}\right\rceil$
        \STATE Run $N$ entangled Hadamard shots with $(m,\vtheta)$ and obtain $Y^{(1)},\dots,Y^{(N)}\in\{\pm1\}$
        \STATE $\hS\gets\dfrac{1}{N}\sum_{j=1}^{N}Y^{(j)}$
        \STATE Output the estimated phase $\htheta\gets\vtheta+\dfrac{1}{m}\arcsin(\hS)$.
    \end{algorithmic}
\end{algorithm}

\section{Adaptive entanglement-assisted Hadamard test}\label{sec:Adaptive entanglement-assisted Hadamard test}

Compared with the standard Hadamard test (Sec.~\ref{sec:phase_estimation_with_HT}), the shot count by entanglement-assisted Hadamard test is reduced by the factor of $m^2$, where $m$ is the amplification (Sec.~\ref{sec:Phase estimation with entanglement-assisted Hadamard test}). However, the feasible amplification is fundamentally limited by the accuracy of the reference phase, making the method inefficient in the high-precision regime. In this section, we propose an adaptive entanglement-assisted Hadamard test that iteratively refines the reference phase, enabling progressively larger amplification as the estimate improves.

Specifically, at the beginning, suppose we have a coarse reference phase $\vtheta_0$, whose distance to the true phase $\theta$ is upper bounded by $|\theta-\vtheta_0|\le \Delta_0$, and the initial amplification is chosen as $m_0 = \lfloor \frac{\chi}{\Delta_0} \rfloor$ with a fixed branch margin $\chi<\pi/2$ and $m_0\Delta_0\le\chi$. Set the target accuracy as $\epsilon_0 = \gamma\Delta_0$, where $\gamma\in(0,1)$ is the contraction factor. Then the estimated quantum phase $\hat{\theta}_0$ can be achieved by conducting Algorithm~\ref{alg:single_round_parallel_hadamard}. In the following $t$-th round, take the estimated phase in the previous round as the new reference phase, i.e., $\vtheta_t:=\hat{\theta}_{t-1}$, and the corresponding reference bound is the target accuracy in previous round $\Delta_t=\epsilon_{t-1}=\gamma^{t}\Delta_0$. Set the amplification as $m_t=\lfloor \frac{\chi}{\Delta_t}\rfloor$, set the target accuracy as $\epsilon_t=\gamma\epsilon_{t-1}$. Conduct the Algorithm~\ref{alg:single_round_parallel_hadamard} to obtain the estimated phase $\hat{\theta}_t$ of the $t$-th round.

The estimated quantum phase of the final round $\hat{\theta}_{T-1}$ is supposed to be the same as the true phase up to a precision $\epsilon$, i.e., $|\hat{\theta}_{T-1}-\theta|\le \gamma^T\Delta_0\le \epsilon$. Therefore, the total number of round $T$ should be
\begin{equation}\label{eq:total_rounds_T}
    T=\left\lceil\log_{1/\gamma}\frac{\Delta_0}{\epsilon}\right\rceil.
\end{equation}
For a total successful probability $1-p_\fail$, assuming the per-round failed probabilities $p_t$ are equal, we set $p_{t}=1-(1-p_\fail)^{1/T}$. The phase estimation with iterative entanglement-assisted Hadamard test in pseudo-code is summarized in Algorithm~\ref{alg:iterative_exact_eigenstate}.

\begin{algorithm}[htbp]
\renewcommand{\algorithmicrequire}{\textbf{Input:}}
\renewcommand{\algorithmicensure}{\textbf{Output:}}
    \caption{Phase estimation with AEHT}
    \label{alg:iterative_exact_eigenstate}
    \begin{algorithmic}[1]
        \REQUIRE Initial reference $\vtheta_0$, initial reference bound $\Delta_0$, target accuracy $\epsilon$, contraction factor $\gamma$, branch margin $\chi$, total failure probability $p_{\fail}$.
        \ENSURE The estimated phase $\hat{\theta}$.
        \STATE $T\gets\left\lceil\log_{1/\gamma}(\Delta_0/\epsilon)\right\rceil$
        \STATE $p_{\mathrm{round}}\gets 1-(1-p_{\fail})^{1/T}$
        \FOR{$t=0$ to $T-1$}
            \STATE The reference bound of $t$-th round $\Delta_t\gets \gamma^{t}\Delta_0$.
            \STATE The amplification $m_t\gets \left\lfloor \frac{\chi}{\Delta_t}\right\rfloor$.
            \STATE The target phase estimation accuracy $\epsilon_t \gets\gamma^{t+1}\Delta_0$.
            \STATE Conduct Algorithm~\ref{alg:single_round_parallel_hadamard} by taking reference phase $\vtheta_t$, amplification $m_t$, reference bound $\Delta_t$, accuracy $\epsilon_t$, and failure probability $p_{\rm round}$ as inputs, and obtain the quantum phase $\hat{\theta}_t$.
            \STATE The reference for the next round is $\vtheta_{t+1}\gets \hat{\theta}_t$.
        \ENDFOR
        \STATE Output the estimated phase $\htheta\gets\htheta_{T-1}$.
    \end{algorithmic}
\end{algorithm}

By Lemma~\ref{lem:parallel_hadamard_shot_count}, it suffices to choose the shot count of the $t$-th round as
\begin{equation}
    N_t
    =
    \left\lceil
    \frac{2}{r_t^2}
    \ln\frac{2}{p_{\mathrm{round}}}
    \right\rceil,
\end{equation}
where
\begin{align*}
    r_t :=\frac{m_t\epsilon_t \left(\sqrt{\cos^2(m_t\Delta_t)+(m_t\epsilon_t)^2}-m_t\epsilon_t\sin(m_t\Delta_t)\right)}{1+(m_t\epsilon_t)^2}.
\end{align*}
Multiply both numerator and denominator by $\sqrt{\cos^2(m_t\Delta_t)+(m_t\epsilon_t)^2}+m_t\epsilon_t\sin(m_t\Delta_t)$, this gives 
\begin{equation}
    r_t = \frac{m_t\epsilon_t\cos^2(m_t\Delta_t)}
    {\sqrt{\cos^2(m_t\Delta_t)+(m_t\epsilon_t)^2}
    +m_t\epsilon_t\sin(m_t\Delta_t)}.
\end{equation}
Using $m_t\Delta_t\le\chi$ and $m_t\epsilon_t\le\gamma\chi$, we obtain the lower bound
\begin{equation}
\label{eq:r_t_lower_bound}
    r_t
    \ge
    \frac{\gamma m_t\Delta_t\cos^2\chi}{C_{\chi,\gamma}},
\end{equation}
where $C_{\chi,\gamma} := \sqrt{1+\gamma^2\chi^2} + \gamma\chi\sin\chi$. Consequently, the shot count of the $t$-th round satisfies
\begin{equation}
    N_t \le 1 + \frac{2C_{\chi,\gamma}^2}{\gamma^2 m_t^2\Delta_t^2\cos^4\chi} \ln\frac{2}{p_{\mathrm{round}}},
\end{equation}
and summing over all $T$ rounds gives
\begin{equation}
\label{eq:adaptive_total_shot_count}
    N_{\mathrm{AEHT}} = \sum_{t=0}^{T-1}N_t 
    \le T + \frac{2C_{\chi,\gamma}^2}{\gamma^2\cos^4\chi} \ln\frac{2}{p_{\mathrm{round}}} \sum_{t=0}^{T-1}\frac{1}{m_t^2\Delta_t^2}.
\end{equation}
For the optimal choice $m_t=\lfloor\chi/\Delta_t\rfloor\approx\chi/\Delta_t$, we have $m_t\Delta_t\approx\chi$, so the per-round shot count is approximately constant:
\begin{equation}
\label{eq:adaptive_round_shot_count}
    N_t \approx \frac{2C_{\chi,\gamma}^2}{\gamma^2\chi^2\cos^4\chi} \ln\frac{2}{p_{\mathrm{round}}}.
\end{equation}
Using $p_{\mathrm{round}}=1-(1-p_{\mathrm{fail}})^{1/T}\approx p_{\mathrm{fail}}/T$, the total shot count becomes
\begin{equation}
    N_{\mathrm{AEHT}} \approx T +   \frac{2C_{\chi,\gamma}^2 T}{\gamma^2\chi^2\cos^4\chi} \ln\frac{2T}{p_{\mathrm{fail}}}.
\end{equation}
By Eq.~\eqref{eq:total_rounds_T}, $T=\mathcal{O}\left(\log\frac{1}{\epsilon}\right)$. Thus $N_{\mathrm{AEHT}}=\mathcal{O}\left(\log\frac1\epsilon\right)$, up to doubly logarithmic factors.

\section{Phase estimation with imperfect state preparation}
\label{sec:imperfect}
In this section, we consider a more practical setting where the prepared eigenstates are imperfect. Let $\ket{\phi_{\mathrm{tar}}}$ be the eigenstate of $U$, with eigenphase $\theta_{\mathrm{tar}}$, and let $\ket{\psi}$ denote the state actually prepared on the quantum device. In general, $\ket{\psi}$ contains components from eigenstates other than $\ket{\phi_{\mathrm{tar}}}$ and is therefore not an exact eigenstate of $U$. Consequently, the measured overlap takes the form
\begin{equation}
\label{eq:effective_overlap}
    \bra{\psi}U\ket{\psi}
    =
    \rho e^{i\theta_{\psi}},
\end{equation}
where $\theta_{\psi}$ is the effective phase associated with the prepared state and $\rho = \lvert\bra{\psi}U\ket{\psi}\rvert \in [0,1]$ is the amplitude, which quantifies the attenuation of the Hadamard-test signal.

We decompose the total error relative to the target eigenphase into the preparation error $\lvert\theta_{\psi}-\theta_{\mathrm{tar}}\rvert$ and the algorithmic error $|\hat{\theta}_\psi -\theta_{\psi}|$ in estimating $\theta_{\psi}$, where $\hat{\theta}_\psi$ refers to the estimated phase. The latter further consists of the \textit{amplitude-induced bias} $\varepsilon_\bias$ and the finite-shot \textit{statistical error} $\varepsilon_\stat$. The preparation error is determined by the quality of $\ket{\psi}$ and cannot be reduced by increasing the shot count or by adaptive rounds. The algorithmic error, by contrast, can be systematically controlled. We first derive bounds on the preparation error and the amplitude-induced bias, and then develop an adaptive entanglement-assisted protocol that repeatedly updates the reference phase and jointly optimizes the amplification level and shot count. The resulting procedure suppresses the amplitude-induced bias round by round and estimates the effective phase $\theta_{\psi}$ to a prescribed accuracy with a success probability.

\subsection{Error analysis}
\label{sec:decomposition}

We now quantify the different error sources introduced above. We first relate the effective phase $\theta_{\psi}$ to the target eigenphase $\theta_{\mathrm{tar}}$ through the infidelity of the prepared state, thereby obtaining an upper bound on the preparation error. We then analyze the EHT signal generated by $\ket{\psi}$ and decompose the algorithmic error in estimating $\theta_{\psi}$ into a finite-shot statistical error and an amplitude-induced bias. This decomposition distinguishes the error intrinsic to imperfect state preparation itself from the errors introduced by the phase-estimation procedure, and paves the path for the adaptive protocol developed in the next subsection.

\subsubsection{Preparation error}
The imperfectly prepared state $\ket{\psi}$ can be decomposed as
\begin{equation}
    \ket{\psi}
    =
    \sqrt{1-\eta}\ket{\phi_{\mathrm{tar}}}
    +
    \sqrt{\eta}\ket{\phi_{\mathrm{tar}}^{\perp}},
\end{equation}
where $\ket{\phi_{\mathrm{tar}}}$ is the target eigenstate and $\ket{\phi_{\mathrm{tar}}^{\perp}}$ is the state perpendicular to it, i.e., $\braket{\phi_{\mathrm{tar}}}{\phi_{\mathrm{tar}}^{\perp}}=0$; here $\eta=1-|\braket{\phi_{\mathrm{tar}}}{\psi}|^2$ is the state-preparation infidelity.

When the state is prepared imperfectly, the measurement does not access the target phase $\theta_{\mathrm{tar}}$ directly; instead, it accesses the effective phase $\theta_\psi$ defined in Eq.~\eqref{eq:effective_overlap}. The gap $|\theta_{\mathrm{tar}}-\theta_\psi|$ is called the \textit{preparation error}. When the infidelity $\eta<1/2$, this error is upper bounded by
\begin{equation*}
    |\theta_{\mathrm{tar}}-\theta_\psi|
    \le
    \arcsin\!\left(\frac{\eta}{1-\eta}\right)
    =
    \eta+\bigO(\eta^2).
\end{equation*}
The derivation of the upper bound is shown in Appendix~\ref{app:upper_bound_of_preparation_error}.

\subsubsection{Algorithmic error}
When the same entanglement-assisted Hadamard test is applied to an imperfectly prepared state $\ket{\psi}$, each shot produces a parity outcome $Y\in\{-1,1\}$. For $N$ independent shots, the sample mean $\hat{S}=\frac{1}{N}\sum_{j=1}^{N}Y^{(j)}$ has expectation
\begin{equation*}
    \mathbb{E}[\hat{S}]
    =\rho^{m}\sin\!\big(m(\theta_\psi-\vtheta)\big),
\end{equation*}
as derived in Appendix~\ref{app:proof_of_parallel_hadamard_shot_count_imperfect}. Compared with the exact-eigenstate case, the signal contains an additional attenuation factor $\rho^m$, and the phase is replaced by the effective phase $\theta_\psi$. If the exact amplitude $\rho$ is known, the effective phase can be estimated by
\begin{equation}
\label{eq:phase_estimation_imperfect}
    \hat{\theta}_{\psi,\rho}
    =\vtheta+\frac{1}{m}\arcsin\!\left(\frac{\hat{S}}{\rho^{m}}\right),
\end{equation}
provided that $|\hat{S}|\le\rho^m$ and the residual phase satisfies $m|\theta_\psi-\vtheta|<\pi/2$. If the state is prepared perfectly, i.e., $\rho=1$ and $\theta_\psi=\theta_{\mathrm{tar}}$, then the estimator reduces to the exact-eigenstate case of Sec.~\ref{sec:Phase estimation with entanglement-assisted Hadamard test}.

In practice, the exact amplitude $\rho$ is generally unknown, so Eq.~\eqref{eq:phase_estimation_imperfect} cannot be applied directly. We instead use the amplitude-unaware estimator
\begin{equation} \label{eq:amplitude_unaware_estimator}
    \hat{\theta}_{\psi}
    =\vtheta+\frac{1}{m}\arcsin(\hat{S}),
\end{equation}
which treats the measured signal as if it had unit amplitude. Ignoring the attenuation factor $\rho^m$ introduces a systematic bias. To distinguish this bias from finite-shot fluctuations, define the infinite-shot limit of the estimator as
\begin{equation*}
    \theta_{\psi,\infty}
    :=\vtheta+\frac{1}{m}\arcsin\!\big(\mathbb{E}[\hat{S}]\big),
\end{equation*}
where $\mathbb{E}[\hat{S}]=\rho^m\sin\!\big(m(\theta_\psi-\vtheta)\big)$. The triangle inequality then bounds the algorithmic error by
\begin{equation}
    |\hat{\theta}_\psi-\theta_\psi|
    \le
    \underbrace{|\hat{\theta}_\psi-\theta_{\psi,\infty}|}
        _{\text{statistical error }\varepsilon_\stat}
    +
    \underbrace{|\theta_{\psi,\infty}-\theta_\psi|}
        _{\text{amplitude-induced bias }\varepsilon_\bias}.
\end{equation}
The statistical error $\veps_\stat$ measures the fluctuation of the finite-shot estimate around its infinite-shot limit, whereas the amplitude-induced bias $\veps_\bias$ measures the deviation of that limit from the effective phase.

Unlike statistical error $\varepsilon_\stat$, this amplitude bias $\varepsilon_\bias$ persists even if the shot count $N\to + \infty$. However, if we suppose the lower bound of amplitude $\rho_0\le\rho$ is available, the corresponding bias can be guaranteed to be upper bounded as shown in Theorem~\ref{thm:single_round_imperfect}.

\begin{theorem}[Amplitude bias error bound]
\label{thm:single_round_imperfect}
Let $U$ be a unitary and let $\ket{\psi}$ be a prepared state with $\langle\psi|U|\psi\rangle=\rho e^{i\theta_\psi}$, where $\theta_\psi$ is the effective phase and $\rho=\lvert\langle\psi|U|\psi\rangle\rvert$ is the signal amplitude. Assume that $\rho$ is lower bounded by a known constant $\rho_0\in[0,1]$, i.e., $\rho\ge\rho_0$. Let $\vtheta$ be a reference phase satisfying $|\theta_\psi-\vtheta|\le\Delta$ for a reference-error bound $\Delta>0$, and choose an integer amplification $m$ and a branch margin $\chi\in(0,\pi/2)$ such that $m\Delta\le\chi$. Then the amplitude-induced bias $\varepsilon_{\mathrm{bias}}$ satisfies
\begin{equation}
    |\varepsilon_{\mathrm{bias}}|
    \le
    \frac{1-\rho_0^{m}}{m}\tan(m\Delta).
\end{equation}
\end{theorem}

The proof is given in Appendix~\ref{app:proof_of_amplitude_error_bound}. In the regime $m\Delta\approx\chi$, the amplitude-induced bias bound scales as $(1-\rho_0^m)/m$, which decreases with $m$; this contrasts with the case of fixed $\Delta$, where the bias
bound increases with $m$ within the admissible branch.

Due to the ambiguity of $\arcsin (\hat S)$, we require the residual phase to lie on an unambiguous branch, i.e., $|\theta-\vartheta|\le\Delta$ with $m\Delta<\pi/2$. The statistical phase error is shown in Eq.~\eqref{eq:paraller_error}. The shot count to obtain the phase within the required statistical error is given in Lemma~\ref{lem:parallel_hadamard_shot_count}.

\subsection{Adaptive entanglement-assisted Hadamard test with imperfect state preparation}

Initially, we have a coarse reference phase $\vtheta_0$, whose distance to the effective phase is upper-bounded by $|\theta_\psi - \vtheta_0|\le\Delta_0$, and the amplitude $\rho$ is lower bounded by some value $\rho\ge \rho_0$. At the $t$-th round, we set the reference phase as the estimate from the previous round, i.e., $\vtheta_t=\hat{\theta}_{\psi, t-1}$. Meanwhile, the corresponding reference bound is the target accuracy of the previous round $\Delta_t=\epsilon_{t-1}$, so that $|\theta_\psi-\vtheta_t|\le \Delta_t$. The goal of this round is to estimate the effective phase up to accuracy $\epsilon_{t}=\gamma\Delta_t$ for some constant $\gamma\in(0,1)$.

In the imperfect state preparation setting, the error contains two components: the bias caused by using only a lower bound $\rho\ge \rho_0$ of amplitude; and the statistical error caused by finite sampling. Therefore, to guarantee the contraction
\begin{equation}
    |\hat{\theta}_{\psi,t} - \theta_\psi|\le\epsilon_t=\gamma\Delta_t,
\end{equation}
we split the round-wise error budget into
\begin{equation*}
    \epsilon_{\stat,t}:=c_{\stat}\Delta_t,
    \qquad
    \epsilon_{\bias,t}:=c_{\bias}\Delta_t,    
\end{equation*}
where $c_\stat>0$ and $c_\bias>0$ are real coefficients, satisfying $c_{\stat}+c_{\bias}=\gamma$. The two parts can be considered as the \textit{budgets} of the statistical error and the amplitude bias error, respectively. The amplitude bias error $\varepsilon_\bias$ and the statistical error $\varepsilon_\stat$ should be smaller than the corresponding budgets:
\begin{equation*}
    |\veps_{\bias,t}|\le\epsilon_{\bias,t},
    \qquad
    |\veps_{\stat,t}|\le\epsilon_{\stat,t}.
\end{equation*}

The amplification $m_t$ is chosen under two constraints. First, as in the perfect-state case, $m_t$ cannot be arbitrarily close to $\frac{\pi}{2\Delta_t}$, which would lead to an infinite shot count. We therefore impose an upper bound $\chi$ such that $m_t\Delta_t\le\chi<\pi/2$. Second, the bias must stay within its budget, i.e., $|\varepsilon_{\bias,t}|\le\epsilon_{\bias,t}$. By Theorem~\ref{thm:single_round_imperfect}, it is sufficient to choose the amplification of the $t$-th round as
\begin{equation}\label{eq:choice_of_mt_imperfect}
    m_{t}=\max \left\{m\, \Big|\, \frac{1-\rho_0^{m}}{m} \tan(m\Delta_t) \le \epsilon_{\bias, t},\, m\in[1, \lfloor\frac{\chi}{\Delta_t}\rfloor ]\right\}.
\end{equation}
The choice of $m_t$ can be efficiently determined by binary search. To guarantee that a feasible amplification exists, the inequality 
\begin{equation}
    (1-\rho_0)\tan(\Delta_t)\le c_{\mathrm{bias}}\Delta_t,
\end{equation}
should be satisfied. For small $\Delta_t$, we have $\tan(\Delta_t)/\Delta_t\approx 1$, and the condition reduces to $c_{\mathrm{bias}}\gtrsim 1-\rho_0$.

The remaining task is to choose enough shots so that the statistical error $\varepsilon_\stat$ does not exceed its budget $\epsilon_\stat$ with success probability at least $1-p_t$. By Lemma~\ref{lem:parallel_hadamard_shot_count}, it suffices to choose
\begin{equation}
    N_t = \left\lceil\frac{2}{r_t^2} \ln\frac{2}{p_t}\right\rceil,
\end{equation}
where $r_t=\frac{m_t\epsilon_{\stat,t}\left(\sqrt{\cos^2(m_t\Delta_t)+m_t^2\epsilon_{\stat,t}^2} - m_t\epsilon_{\stat, t} \sin(m_t\Delta_t)\right)}{1+m_t^2\epsilon_{\stat,t}^2}$. The shot count depends on $m_t$, so $m_t$ should be chosen to minimize it subject to the feasibility constraints.

For convenience, we set the failure probability to be the same in each round, i.e., $p_{\mathrm{round}}=p_t$ for all $t=1,\dots,T$. Thus, for a given total success probability $1-p_{\fail}$, one can easily derive that
\begin{equation*}
    p_{\mathrm{round}}:=1-(1-p_{\fail})^{1/T}.
\end{equation*}
Since we aim to estimate the effective phase $\theta_\psi$ up to a target accuracy $\epsilon$, we require $|\hat{\theta}_{\psi,T-1}-\theta_\psi| \le \epsilon_{T-1} =\gamma^T\Delta_0\le\epsilon$, where $T-1$ is the final round. Thus the total number of rounds should be
\begin{equation}
    T=\left\lceil\log_{1/\gamma}\frac{\Delta_0}{\epsilon}\right\rceil.
\end{equation}
We summarize the adaptive procedure in Algorithm~\ref{alg:iter_phase_estimation_main}.

\begin{algorithm}[h]
\renewcommand{\algorithmicrequire}{\textbf{Input:}}
\renewcommand{\algorithmicensure}{\textbf{Output:}}
    \caption{Adaptive algorithm for phase estimation with imperfect state preparation }
    \label{alg:iter_phase_estimation_main}
    \begin{algorithmic}[1]
        \REQUIRE Initial reference $\vtheta_0$, initial effective-phase bound $\Delta_0$, target accuracy $\epsilon>0$, total failure probability $p_{\fail}\in(0,1)$, contraction $\gamma\in(0,1)$, coefficients $c_\stat$ and $c_\bias$, branch margin $\chi\in(0,\pi/2)$, amplitude lower bound $\rho_0$.
        \ENSURE The estimated phase $\hat{\theta}_\psi$.
        \STATE $T\gets\left\lceil \log_{1/\gamma}(\Delta_0/\epsilon)\right\rceil$
        \STATE $p_{\mathrm{round}}\gets 1-(1-p_{\fail})^{1/T}$
        \FOR{$t=0$ to $T-1$}
            \STATE The reference bound of the $t$-th round $\Delta_t\gets \gamma^{t} \Delta_0$.
            \STATE The target phase estimation accuracy of $t$-th round $\epsilon_t \gets \gamma^{t+1}\Delta_0$.
            \STATE Calculate the budgets for bias error and statistical error $\epsilon_{\bias,t}\gets c_{\bias}\Delta_t$, \quad $\epsilon_{\stat,t}\gets c_{\stat}\Delta_t$
            \STATE Determine the amplification $m_t\gets$ Eq.~\eqref{eq:choice_of_mt_imperfect}.
            \STATE Obtain $\hat{\theta}_{\psi,t}\gets$ Algorithm~\ref{alg:single_round_parallel_hadamard} with reference $\vtheta_t$, amplification $m_t$, reference bound $\Delta_t$, accuracy $\epsilon_{\stat,t}$, and failure probability $p_{\mathrm{round}}$.
            \STATE The reference for the next round $\vtheta_{t+1}\gets \hat{\theta}_{\psi,t}$.
        \ENDFOR
        \STATE Output the estimation phase  $\htheta_\psi\gets\htheta_{\psi,T-1}$.
    \end{algorithmic}
\end{algorithm}

For a fixed reference phase $\vtheta_t=\vtheta_0$ and a fixed amplification $m_t=m_0$, the amplitude mismatch produces a bias that cannot be removed by increasing the number of shots: additional measurements only suppress the statistical fluctuation around the bias. Consequently, the EHT generally exhibits an error floor determined by the imperfect state preparation and the unknown amplitude. In contrast, the AEHT repeatedly updates the reference phase and shrinks the error, so that the amplitude-induced algorithmic bias is forced to remain within the prescribed round-wise bias budget. Therefore, the AEHT separates the irreducible preparation error, which is fixed by the prepared state $|\psi\rangle$, from the reducible algorithmic error, which can be systematically suppressed by the iterative procedure.

\section{Comparison of Hardware Restart Time}\label{sec:comarison}

In previous sections, the cost of phase estimation was quantified by the number of shots. However, on large-scale quantum processors, this metric does not fully capture the practical overhead. With a growing number of available qubits, many independent shots can be executed in parallel during a single experimental cycle. The limiting resource therefore shifts from the total number of shots to the number of times the full device must be initialized, executed, measured, and reset, referred to as the \textit{device-restart count}. Such restarts are particularly costly on platforms like neutral-atom arrays, where atom loading and rearrangement introduce substantial latency \cite{bluvstein2024logical, henriet2020quantum, saffman2016quantum, ebadi2021quantum}. In this section, we provide a theoretical comparison of the restart overhead between the fixed-$m$ entanglement-assisted Hadamard test and the proposed adaptive protocol.

Suppose the quantum processor comprises $Q$ qubits. For an amplification factor $m$, the entanglement-assisted Hadamard test circuit uses an $m$-qubit GHZ ancilla and $mn_\sys$ system-register copies, occupying a total of $m(n_\sys+1)$ qubits, where $n_\sys$ refers to the number of system qubits. Within one device restart, the number of independent circuit instances that can be executed in parallel is therefore
\begin{equation}
\kappa(m) = \left\lfloor \frac{Q}{m(n_\sys+1)} \right\rfloor .
\end{equation}
Denote the shot count for $t$-th round as $N_t$, then the corresponding number of device restarts is
\begin{equation}
R = \sum_{t} \left\lceil \frac{N_t}{\kappa(m_t)} \right\rceil,
\label{eq:restart_metric}
\end{equation}
where the sum runs over all rounds in the protocol. For the purpose of asymptotic analysis, we neglect the ceiling and floor functions, which affect only lower-order factors and do not alter the scaling with respect to the target precision $\epsilon$.

Consider the EHT with a fixed amplification $m_0=\lfloor\chi/\Delta_0\rfloor$. Its shot count is given by Eq.~\eqref{eq:N_EHT} with $m=m_0$. Since each restart packs $\kappa(m_0)=\lfloor Q/[m_0(n_\sys+1)]\rfloor$ shots, the corresponding restart count is
\begin{equation}
R_{\mathrm{EHT}}
\approx
\frac{m_0(n_\sys+1)}{Q}N_{\mathrm{EHT}}
\approx
\frac{2(n_\sys+1)\Delta_0}{Q\chi\cos^2\chi}
\ln\frac{2}{p_{\mathrm{fail}}}
\cdot
\frac{1}{\epsilon^2},
\label{eq:R_fixed}
\end{equation}
where we used $m_0\Delta_0\approx\chi$. Thus, the fixed-$m$ EHT exhibits a \textit{quadratic} growth in restarts as the required precision tightens, i.e., $R_{\mathrm{EHT}}=\mathcal{O}(1/\epsilon^2)$.

Consider the proposed AEHT scheme, the estimation proceeds over
\begin{equation}
T = \left\lceil \log_{1/\gamma} \frac{\Delta_0}{\epsilon} \right\rceil
\end{equation}
rounds, with $\Delta_t=\gamma^t\Delta_0$, $m_t=\lfloor\chi/\Delta_t\rfloor$, and $\epsilon_t=\gamma\Delta_t$. From Eq.~\eqref{eq:adaptive_round_shot_count}, the per-round shot count is approximately constant,
\begin{equation}
N_t\approx N_{\mathrm{round}}:=\frac{2C_{\chi,\gamma}^2}{\gamma^2\chi^2\cos^4\chi}\ln\frac{2}{p_{\mathrm{round}}},
\end{equation}
where $C_{\chi,\gamma}$ is defined in Sec.~\ref{sec:Adaptive entanglement-assisted Hadamard test} and $p_{\mathrm{round}}=1-(1-p_{\mathrm{fail}})^{1/T}$. For large $T$, $p_{\mathrm{round}}\approx p_{\mathrm{fail}}/T$, so $\ln(2/p_{\mathrm{round}})$ introduces only a doubly logarithmic dependence on $\epsilon$.

The number of restarts required in round $t$ is
\begin{equation}
R_t \approx \frac{(n_\sys+1)m_t}{Q}N_t \approx \frac{(n_\sys+1)\chi N_{\mathrm{round}}}{Q\Delta_0}\gamma^{-t}.
\end{equation}
Summing over all $T$ rounds, the total restart count for the adaptive protocol is
\begin{align}
R_{\mathrm{AEHT}} &= \sum_{t=0}^{T-1}R_t \approx \frac{(n_\sys+1)\chi N_{\mathrm{round}}}{Q\Delta_0} \sum_{t=0}^{T-1}\gamma^{-t} \nonumber\\
&\approx \frac{(n_\sys+1)\chi N_{\mathrm{round}}}{Q\Delta_0} \cdot \frac{\gamma\Delta_0}{(1-\gamma)\epsilon} \nonumber\\
&= \frac{2(n_\sys+1)C_{\chi,\gamma}^2}{Q(1-\gamma)\gamma\chi\cos^4\chi} \ln\frac{2}{p_{\mathrm{round}}} \cdot \frac{1}{\epsilon}.
\label{eq:R_adapt}
\end{align}
The proposed AEHT therefore yields a restart count that grows only as $\mathcal{O}(1/\epsilon)$, up to doubly logarithmic corrections.

\subsection{Advantage of the AEHT}

Comparing the leading-order scaling of Eqs.~\eqref{eq:R_fixed} and \eqref{eq:R_adapt}, we obtain
\begin{align}
\frac{R_{\mathrm{AEHT}}}{R_{\mathrm{EHT}}} &\approx \frac{C_{\chi,\gamma}^2}{(1-\gamma)\gamma\cos^2\chi} \cdot \frac{\ln(2/p_{\mathrm{round}})}{\ln(2/p_{\mathrm{fail}})} \cdot \frac{\epsilon}{\Delta_0} \nonumber\\
&\approx \frac{C_{\chi,\gamma}^2}{(1-\gamma)\gamma\cos^2\chi} \left(1+\frac{\ln T}{\ln(2/p_{\mathrm{fail}})}\right) \frac{\epsilon}{\Delta_0}.
\label{eq:ratio}
\end{align}
Ignoring constants and using $T=\mathcal{O}(\log(1/\epsilon))$, we obtain
\begin{equation}
\frac{R_{\mathrm{AEHT}}}{R_{\mathrm{EHT}}} = \mathcal{O}\!\left(\epsilon\cdot\log\log\frac{1}{\epsilon}\right).
\end{equation}
For high-precision estimation where $1\gg \epsilon$, the factor $\epsilon$ dominates the ratio, while the logarithmic term grows only as $\log \log(1/\epsilon)$. 
This demonstrates that the AEHT method achieves a superlinear reduction in device-restart count compared with the fixed-$m$ EHT approach in the high-precision regime. 

\subsection{Hardware-constraint amplification}
\label{sec:hardware_refinements}

In reality, the number of qubits on the quantum processors are finite $Q<+\infty$, implying the choice of amplification $m$ is also limited by the capability of quantum hardware, i.e.,
\begin{equation}
    m\le m_{\hw}:=\left\lfloor \frac{Q}{n_{\sys}+1}\right\rfloor
\end{equation}
This is the largest amplification $m$ for which an amplification-$m$ entanglement-assisted Hadamard-test circuit fits within the $Q$ available qubits. In Algorithm~\ref{alg:iterative_exact_eigenstate}, the ideal amplification is $m_t^{\rm ideal}=\lfloor\chi/\Delta_t\rfloor$, while in Algorithm~\ref{alg:iter_phase_estimation_main} the ideal amplification $m_t^{\rm ideal}$ is given by Eq.~\eqref{eq:choice_of_mt_imperfect}. In both cases, the physical qubit number is incorporated by comparing the ideal value with $Q/(n_{\sys}+1)$ and keeping the smaller one:
\begin{equation}
    m_t=\min\left\{m_t^{\rm ideal}, m_{\hw}\right\}.
\end{equation}
To choose the amplification that corresponds to the fewest shot count, we design an optimization algorithm. The detailed derivation and the algorithm in pseudo-code are shown in Appendix~\ref{app:Optimal_amplification_with_imperfect_state_preparation}. When the reference error $\Delta_t$ is still large, the ideal amplification lies below $m_{\hw}$, and the two algorithms reduce to their earlier forms. When $\Delta_t$ becomes sufficiently small, the amplification saturates at $m_{\hw}$, and further precision is gained by increasing the shot count rather than by enlarging the GHZ state.

Once the amplification saturates, both schemes operate with a fixed, bounded amplification: the fixed-$m$ EHT method uses $m_0$, while the adaptive method saturates at $m_{\hw}$. To see the resulting scaling, consider a saturated round with $m_t=m_{\hw}$. By Lemma~\ref{lem:parallel_hadamard_shot_count} and $\epsilon_t=\gamma\Delta_t$, its shot count is
\begin{equation}
    N_t
    \approx
    \frac{2}{m_{\hw}^2\gamma^2\Delta_t^2\cos^2(m_{\hw}\Delta_t)}
    \ln\frac{2}{p_{\mathrm{round}}},
\end{equation}
so its restart count is
\begin{align}
    R_t &\approx \frac{(n_{\sys}+1)m_{\hw}}{Q}N_t\\
    &\approx \frac{2(n_{\sys}+1)}{Q m_{\hw}\gamma^2\Delta_t^2\cos^2(m_{\hw}\Delta_t)} \ln\frac{2}{p_{\mathrm{round}}}.
\end{align}
    
Since $\Delta_t=\gamma^t\Delta_0$, the sum over the saturated rounds is dominated by the final round and scales as $1/\epsilon^2$, which is
\begin{equation}
    R_{\mathrm{AEHT}}^{\mathrm{sat}} \approx \frac{1}{m_{\hw}\epsilon^2}.
\end{equation}
The fixed-$m$ restart count in Eq.~\eqref{eq:R_fixed} similarly scales as $R_{\mathrm{EHT}}\sim 1/(m_0\epsilon^2)$. Ignoring constant factors, the ratio is therefore
\begin{equation}
    \frac{R_{\mathrm{AEHT}}^{\mathrm{sat}}}{R_{\mathrm{EHT}}}\approx \frac{m_0}{m_{\hw}} < 1,
\end{equation}
because saturation requires $m_{\hw}>m_0$. Hence the AEHT still has a constant-factor advantage over the fixed-$m$ EHT in this regime.

\section{Numerical experiments}
\label{sec:numerical_experiments}

We now numerically demonstrate the effectiveness of the adaptive protocol developed in this work. For concreteness, we encode the target phase $\theta_\tar=2.0$ into a single-qubit phase unitary
\begin{equation*}
    U=\operatorname{diag}(1,e^{\ri\theta_{\tar}}).
\end{equation*}
The corresponding eigenstate is simply $\ket{1}$. All experiments use the initial reference $\vtheta_0=1.8$, the initial bound $\Delta_0=0.2$, the branch margin $\chi=\pi/4$, and total success probability $1-p_\fail=0.95$.

To compare the device-restart count, we model a device with $Q=2500$ qubits and a one-qubit system register, and the corresponding hardware constraint amplification is $m_\hw=1250$. An amplification-$m$ EHT circuit occupies $2m$ qubits, so the number of independent instances that fit in one initialization is
\begin{equation*}
    \kappa(m)=\left\lfloor\frac{2500}{2m}\right\rfloor.
\end{equation*}
If round $t$ uses $N_t$ shots at amplification $m_t$, its restart cost is $\lceil N_t/\kappa(m_t)\rceil$, and the schedule cost is
\begin{equation} \label{eq:numerical_restart_cost}
    R=\sum_t\left\lceil\frac{N_t}{\kappa(m_t)}\right\rceil.
\end{equation}

In the following, we compare the hardware restart counts of the proposed adaptive-$m$ method and fixed-$m$ method. However, from Algorithm~\ref{alg:iter_phase_estimation_main}, the adaptive-$m$ method may certify a precision better than the requested target $\epsilon$, thereby paying more cost. To make a fair comparison, we replace the final-round precision by
\begin{equation}
    \epsilon_{T-1}
    =
    \max\left\{\epsilon,\gamma^{T}\Delta_0\right\}
\end{equation}
so that the last round does not overshoot the requested accuracy. The selection of the amplification factors uses the Algorithm~\ref{alg:min_restart_amplification} in Appendix~\ref{app:Optimal_amplification_with_imperfect_state_preparation}.

\subsection{Restart reduction with perfect eigenstate preparation}
\label{sec:numerical_exact_restarts}
Suppose quantum circuits are noiseless and we can perfectly prepare the eigenstate $\ket{1}$. For different target accuracies, we compare the required device-restart count of the SHT method (pink-diamond curve), EHT method with a $3$-qubit GHZ entangled state (green-triangle curve), and the proposed AEHT method (blue-circle curve), as shown in Fig.~\ref{fig:num_exact_restarts}. 
The three methods use the same hardware budget of $Q=2500$ qubits, allowing up to $1250$ SHT shots, $416$ EHT shots, or $\lfloor1250/m_t\rfloor$ AEHT shots in round $t$ to run in parallel per restart. The device-restart counts therefore account for the parallel capacity of each method under the same hardware constraint.

All three restart counts increase as the target accuracy becomes more stringent. Before the hardware limit is reached, the AEHT restart count grows much more slowly than those of the other two methods, and its advantage in device-restart count widens. Once the amplification saturates, the AEHT restart count grows faster than before, but it still retains an advantage of roughly two orders of magnitude over the other two methods.
The numerical results agree with our analysis as shown in Sec.~\ref{sec:comarison}. The AEHT method unlocks the amplification power of large entangled states, and the maximum amplification of the AEHT method increases from $m=11$ at target accuracy $\epsilon=10^{-2}$ to $m=1250$ at target accuracy $\epsilon=10^{-4}$.

\begin{figure}[h]
    \centering
    \includegraphics[width=\columnwidth]{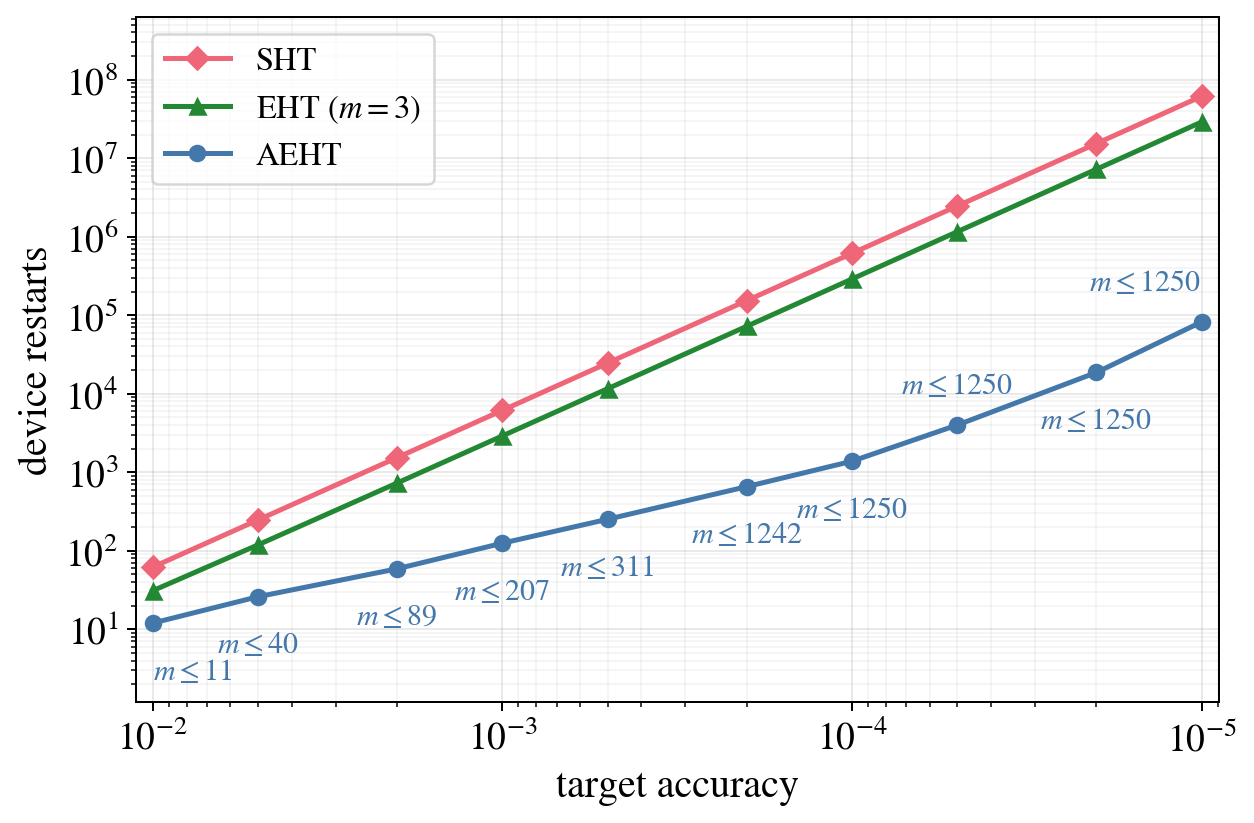}
    \caption{Device-restart count for exact eigenstate preparation. The labels on the AEHT refer to the maximum amplification used.}
    \label{fig:num_exact_restarts}
\end{figure}

We next focus on the target accuracy $\epsilon=10^{-3}$. The amplification $m_t$ and estimated phase error $\epsilon_t$ with respect to the $t$-th round are shown in Fig.~\ref{fig:num_exact_schedule}. In this case, the number of rounds is $T=6$. In the first round, due to the limitation of the reference phase, the maximal amplification can only be $m_0=2$. As $t$ increases, the estimated phase error (blue-circle curve) decreases, while the amplification $m_t$ (orange curve) grows. Eventually, in the last round, the estimated phase error falls below the target accuracy (black dotted line), and the corresponding amplification becomes $m_5=207$. 
\begin{figure}[h]
    \centering
    \includegraphics[width=\columnwidth]{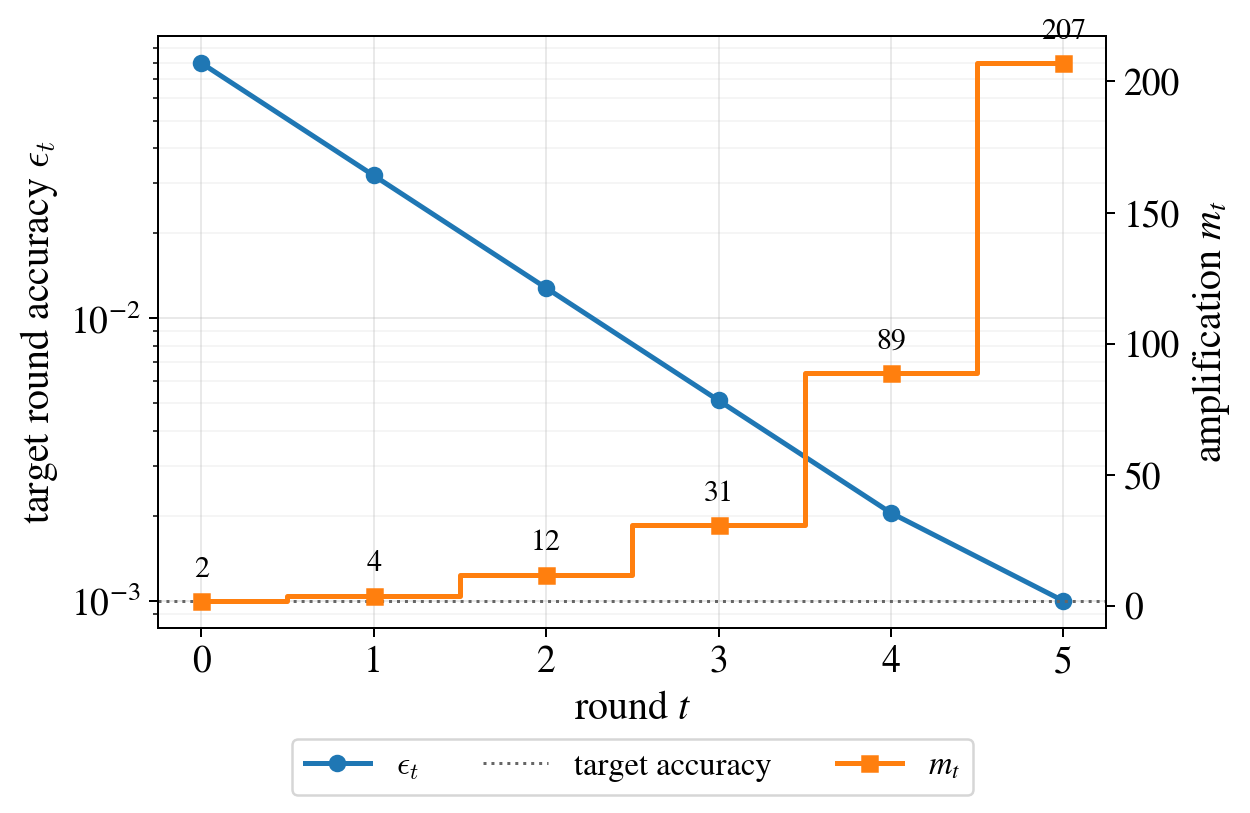}
    \caption{The amplification $m_t$ and estimated phase error $\epsilon_t$ with respect to the increase of round $t$ in the Algorithm~\ref{alg:iterative_exact_eigenstate}. The dotted line refers to the target accuracy $\epsilon=10^{-3}$.}
    \label{fig:num_exact_schedule}
\end{figure}

\subsection{Restart reduction with imperfect eigenstate preparation}
\label{sec:numerical_imperfect_restarts}
In reality, it is generally hard to prepare the desired state perfectly. Instead of the eigenstate $\ket{1}$, what we actually prepare is
\begin{equation*}
    \ket{\psi}=\sqrt{\eta}\ket{0}+\sqrt{1-\eta}\ket{1}, \qquad \eta=0.01.
\end{equation*}
The corresponding overlap is
\begin{equation*}
    \bra{\psi}U\ket{\psi}=\eta+(1-\eta)e^{\ri\theta_{\tar}}:=\rho e^{\ri\theta_\psi},
\end{equation*}
When conducting experiments, the exact values are generally unknown, so we assume the amplitude is lower bounded by $\rho_0=1-2\eta=0.98$. As introduced in Section~\ref{sec:decomposition}, the error contains three parts: preparation error, statistical error, and amplitude bias error. The preparation error in this case is $|\theta_\psi-\theta_{\tar}|$, which is irreducible. Thus, in this section, the accuracy $\epsilon$ refers to the algorithmic error $|\hat{\theta}_\psi-\theta_\psi|$, where $\hat{\theta}_\psi$ refers to the estimated phase, and the preparation error is excluded. The statistical error vanishes as the device-restart count increases. 

For a fixed-$m$ EHT method, the amplitude bias error cannot be mitigated, no matter how many device-restart counts we take. The comparison of device-restart count over different target accuracy is shown in Fig.~\ref{fig:num_imperfect_restarts}. For the SHT method (pink-diamond curve) and the fixed-$m$ EHT method (green-triangle curve) with $m=2$, as the target effective-phase accuracy $\epsilon$ approaches the bias bound (dotted curve and dash-dot curve, respectively), the device-restart counts increase dramatically. For the two conventional methods, they can only approach the bias bound and cannot achieve an accuracy better than the bias bound. In contrast, the AEHT method is not limited by the bias bound and can estimate the quantum phase more accurately than the conventional methods with fewer device-restart counts.

\begin{figure}[h]
    \centering
    \includegraphics[width=\columnwidth]{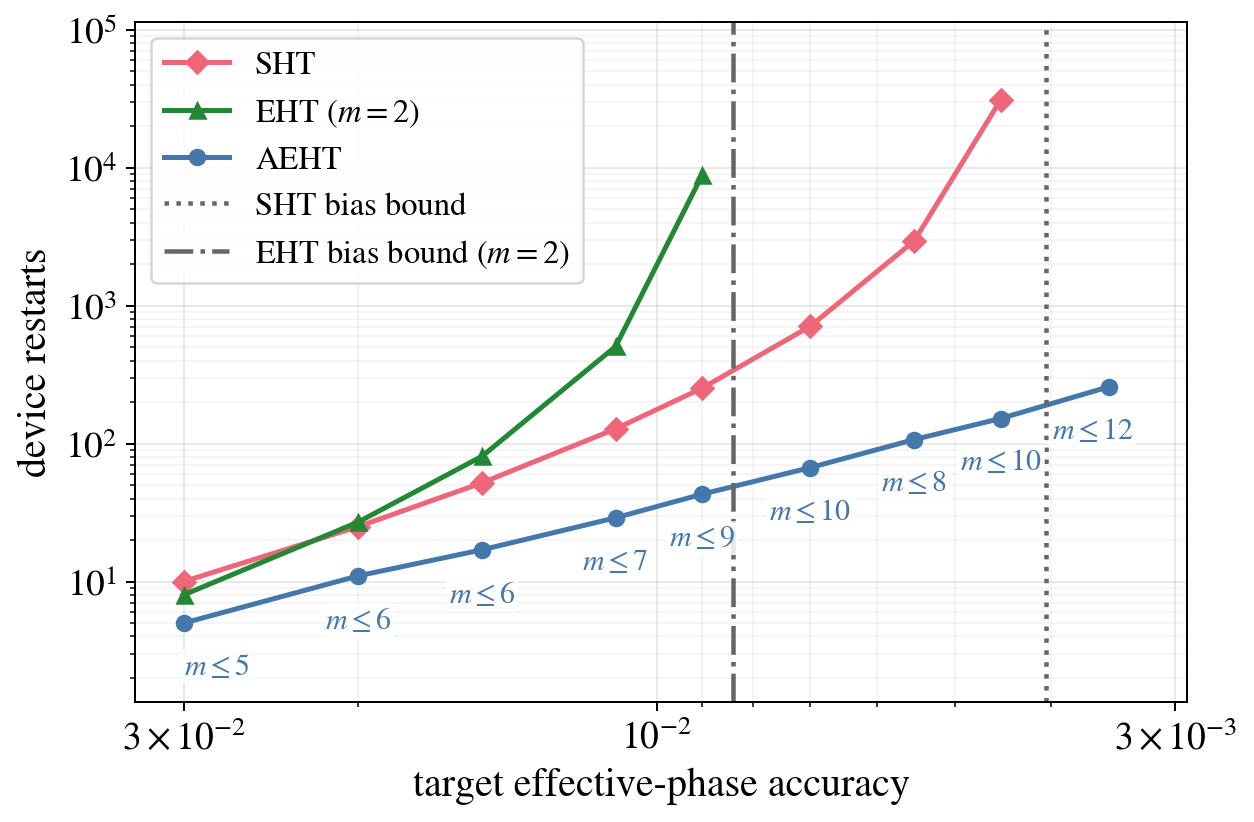}
    \caption{Device-restart count for imperfect eigenstate preparation. The vertical dotted curve and dash-dot curve mark the bias bounds of SHT and Fixed-$m$ EHT ($m=2$), respectively.}
    \label{fig:num_imperfect_restarts}
\end{figure}

\subsection{Circuit-noise robustness}
\label{sec:noisy_simulation}

In previous sections, we studied the upper bound of the device restart count to achieve target phase accuracy. Here, we demonstrate the robustness of the proposed AEHT method by simulating noisy quantum circuits. We suppose the prepared quantum states are imperfect as shown in Sec.~\ref{sec:numerical_imperfect_restarts}. When simulating the quantum circuit, we also assume the quantum gates are corrupted by a depolarizing channel $\cN_p$, which maps the input state toward the maximally mixed state with probability $p$, i.e.,
\begin{equation}
    \cN_p(\psi) = (1-p)\psi + p\frac{I}{d},
\end{equation}
where $p$ refers to the noise level and $I/d$ refers to the maximally mixed state. In the simulation, we set the single-qubit and two-qubit gates to be corrupted by single-qubit and two-qubit depolarizing channels with probabilities $p_1=10^{-4}$ and $p_2=10^{-3}$, respectively. Moreover, the readout error is set to $p_{\rm ro}=10^{-2}$. The numerical simulations were implemented using Qiskit~\cite{javadi2024quantum} version 2.2.3 and executed on a CPU-based server equipped with Intel Xeon Gold 6226R processors operating at 2.90~GHz, with 64 total hardware threads and 1,006.6~GiB of available RAM. No GPU acceleration was used.

In the circuit-sampled numerical experiments, we set hardware to contain $Q=2500$ qubits, and due to the classical memory limitation of the server we used, we set the amplification $m\le 8$ .
The achievable accuracy for different device-restart counts is shown in Fig.~\ref{fig:num_noisy_imperfect}. The dashed pink curve with hollow diamonds and the dashed blue curve with hollow circles denote the theoretical upper bounds on the achievable accuracy of the SHT and AEHT methods under imperfect state preparation, respectively. As the device-restart count increases, the SHT bound saturates at the bias limit indicated by the dotted horizontal line, whereas the AEHT bound is not limited by such a bias floor and keeps improving; this behavior is also shown in Fig.~\ref{fig:num_imperfect_restarts}. The dashed pink curve with solid diamonds and the dashed blue curve with solid circles show the actual errors of the SHT and AEHT methods obtained from sampling noisy circuits. As the device-restart count increases, the actual SHT error converges to a worse value because of circuit noise. In contrast, the actual AEHT error follows the trend of its upper bound and produces an increasingly accurate phase estimate. These numerical experiments demonstrate the robustness of the proposed AEHT method.

\begin{figure}[h]
    \centering
    \includegraphics[width=\columnwidth]{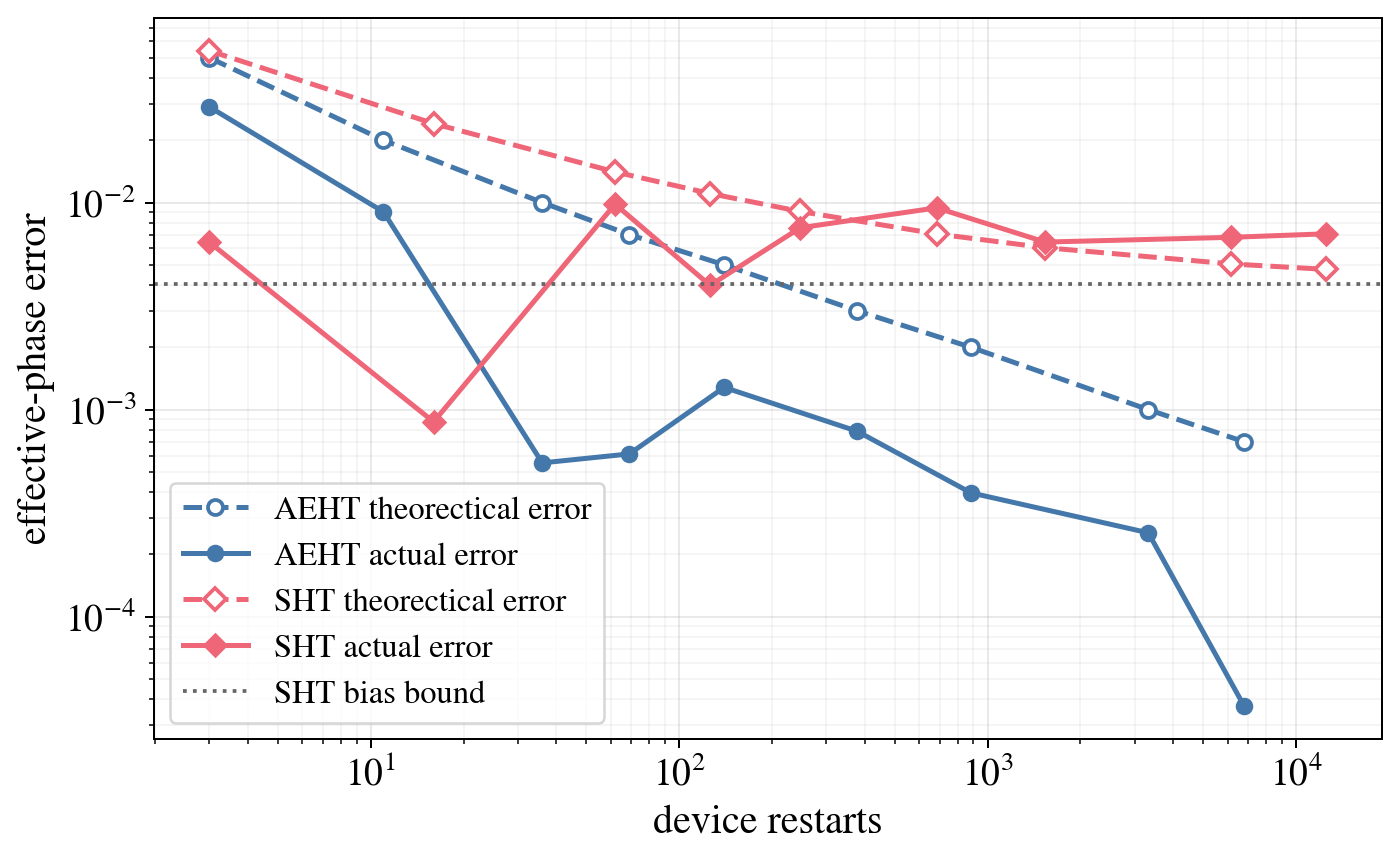}
    \caption{Circuit-sampled effective accuracy versus device restarts under the same noise model for imperfect eigenstate preparation. For each target $\epsilon$, the plot shows the requested target accuracy and the realized noisy effective-phase error at the corresponding restart budget.}
    \label{fig:num_noisy_imperfect}
\end{figure}

\section{Conclusion and discussion}\label{sec:conclusion}

In this work, we proposed an AEHT for high-precision quantum-phase estimation. By iteratively refining the reference phase, the protocol permits progressively larger GHZ-state amplification and overcomes the reference-accuracy limitation of conventional fixed-amplification EHT methods. For exact eigenstates, the resulting shot complexity is $\mathcal{O}\left(\log\frac1\epsilon\right)$ up to doubly logarithmic factors when the amplification does not reach the hardware limit. 
We further extended the method to imperfect eigenstate preparation. By separating the irreducible preparation error from the statistical error and amplitude-induced bias, the adaptive protocol systematically suppresses the algorithmic error, whereas a conventional fixed-amplification test generally exhibits a nonvanishing bias floor.

To better characterize practical cost, we introduced the device-restart count, which accounts for parallel circuit execution under a finite qubit budget. In the hardware-unconstrained regime, the adaptive protocol achieves a restart complexity of $\cO(1/\epsilon)$, up to logarithmic factors, compared with $\cO(1/\epsilon^2)$ for fixed amplification. When the amplification reaches the hardware limit, the quadratic scaling is recovered, while a constant-factor advantage remains. Numerical results support these conclusions for both exact and imperfect eigenstates and demonstrate robustness under representative circuit noise.

The proposed method provides a hardware-aware route to precision phase estimation with substantially reduced device-restart overhead on near-term quantum processors. It would be interesting to demonstrate the proposed AEHT method on quantum hardware, such as neutral atom platform, in the future.

\textbf{Code availability.--}
The code has been uploaded to Github, which can be found at {\url{https://github.com/HengzhunChen/EntangleHT}}.

\textbf{Acknowledgements.--}
B.Zhao would like to thank Xuanqiang Zhao, Ge Bai and Qiushi Liu for their insightful discussions. This work was supported in part by the National Natural Science Foundation of China under Grant No. 12271109 and 12526211; by the Shanghai Pilot Program for Basic Research-Fudan University under Grant No. 21TQ1400100 (22TQ017); by the Scientific Research Innovation Capability Support Project for Young Faculty under Grant No. SRICSPYF-ZY2025159; and by the Xuemin Institute of Advanced Studies, Fudan University.

\bibliography{PRX_Q_ref.bib}


\vspace{2cm}
\onecolumngrid
\vspace{2cm}
\begin{center}
{\textbf{\large Appendix for Efficient phase estimation with adaptive entanglement-assisted Hadamard test}}
\end{center}

\appendix

\setcounter{proposition}{0}
\renewcommand{\theproposition}{S\arabic{proposition}}

\section{Conventional quantum phase estimation} \label{app:quantum_phase_estimation}

Quantum phase estimation~\cite{kitaev1995quantum} addresses the following task: given a unitary operator $U$ and one of its eigenstates $\ket{\psi}$ satisfying $U\ket{\psi} = e^{2\pi i \theta}\ket{\psi}$, estimate the unknown phase $\theta\in[0,1)$ to $n$ bits of precision.

The standard formulation employs two registers: an $n$-qubit ancilla register initialized to $\ket{0}^{\ox n}$ and a target register prepared in $\ket{\psi}$. After applying Hadamard gates to all ancilla qubits, each ancilla qubit $j$ controls an application of $U^{2^j}$ to the target register. By the phase kickback mechanism, the ancilla register is transformed into the state $\frac{1}{2^{n/2}} \sum_{k=0}^{2^n -1} e^{2\pi ik\theta} \ket{k}$, which is precisely the quantum Fourier transform of the $n$-bit integer closest to $2^n\theta$. Applying the inverse quantum Fourier transform ($QFT^{-1}$) to the ancilla register and measuring in the computational basis then yields an $n$-bit approximation of $\theta=0.\theta_1\cdots \theta_n$. The circuits for QPE is shown in Fig.~\ref{fig:conventional_QPE}.

\begin{figure}[h]
    \centering
    \includegraphics[width=0.5\linewidth]{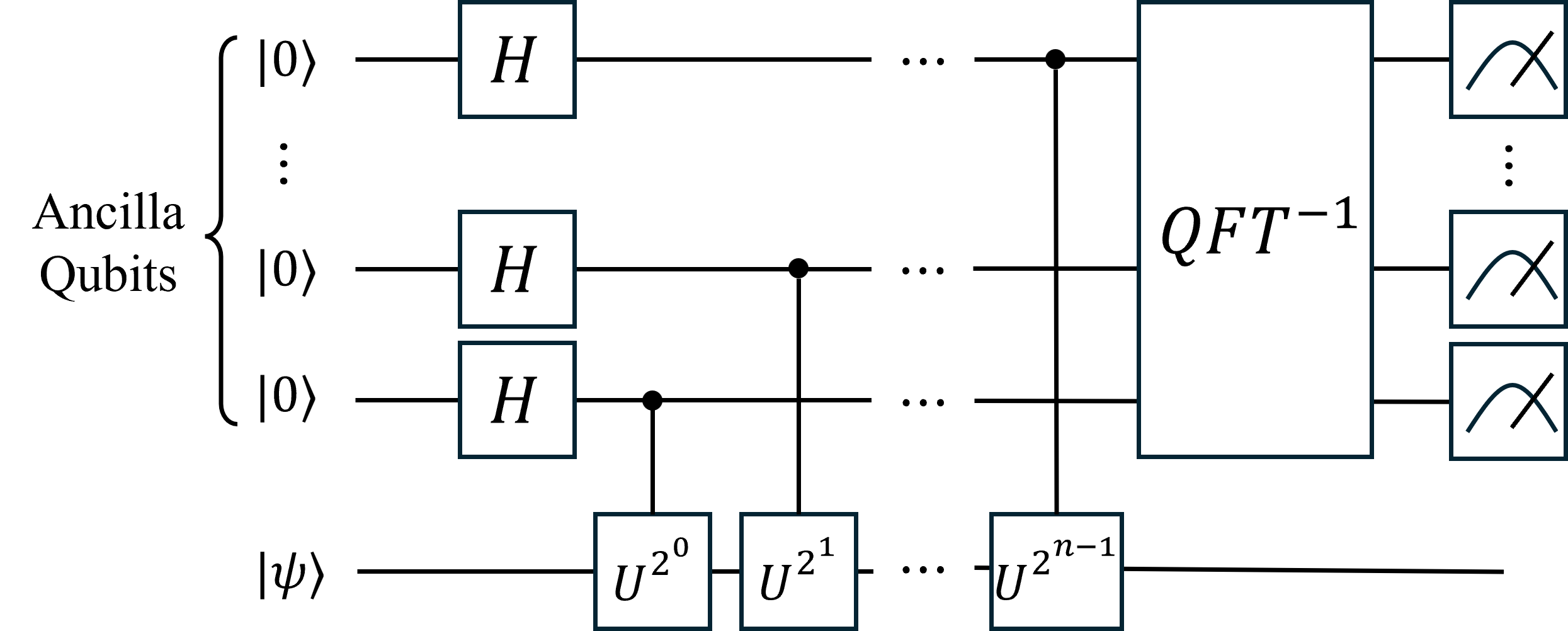}
    \caption{Quantum circuit for QPE.}
    \label{fig:conventional_QPE}
\end{figure}

Obviously, the depth of the quantum circuit for QPE is large. Considering that quantum noise on near-term quantum devices is unavoidable, the noise accumulates with the depth of the circuit, making the measurement results unreliable. Therefore, such a method for estimating quantum phase is not practical for near-term quantum devices.



\section{Expectation values of standard and entanglement-assisted Hadamard tests}
\label{app:hadamard_expectation_values}

The standard Hadamard test is recovered from the entanglement-assisted construction by setting $m=1$. We therefore derive the expectation value for general $m\ge1$, covering both tests in a single proof.

\begin{lemma}\label{lem:hadamard_expectation_values}
    Assume $U\ket{\phi}=e^{\ri\theta}\ket{\phi}$. For any integer $m\ge 1$, the $m$-copy entanglement-assisted Hadamard phase estimation circuit with reference phase $\vtheta$ and parity outcome $Y=(-1)^{|z|}$ gives
    \begin{equation*}
        \bbE[Y]=\sin(m(\theta-\vtheta)).
    \end{equation*}
\end{lemma}

\begin{proof}
    Prepare $\ket{\GHZ_m}$ on the ancillas and $\ket{\phi}^{\otimes m}$
    on the system registers.  The initial state is
    \begin{equation*}
        \ket{\Phi_0}
        =
        \frac{1}{\sqrt{2}}
        \left(\ket{0}^{\otimes m}+\ket{1}^{\otimes m}\right)
        \ket{\phi}^{\otimes m}.
    \end{equation*}
    After the parallel controlled-$U$ layer, the eigenvalue equation
    $U\ket{\phi}=e^{\ri\theta}\ket{\phi}$ gives
    \begin{equation*}
        \ket{\Phi_1}= \frac{1}{\sqrt{2}}
        \left(\ket{0}^{\otimes m} + e^{\ri m\theta}\ket{1}^{\otimes m}\right) \ket{\phi}^{\otimes m}.
    \end{equation*}
    The reference phase block maps $\ket{1}^{\otimes m}\rightarrow -\ri
    e^{-\ri m\vtheta}\ket{1}^{\otimes m}$, so the state before the final
    Hadamards is
    \begin{equation*}
        \ket{\Phi_2}
        =
        \frac{1}{\sqrt{2}}
        \left(
            \ket{0}^{\otimes m}
            +
            (-\ri)e^{\ri m(\theta-\vtheta)}
            \ket{1}^{\otimes m}
        \right)
        \ket{\phi}^{\otimes m}.
    \end{equation*}
    We use the standard identities
    \begin{equation} \label{eq:multi_hadamard_relation}
        H^{\otimes m}\ket{0}^{\otimes m}
        =
        \frac{1}{\sqrt{2^m}}\sum_{z\in\{0,1\}^m}\ket{z},
        \qquad
        H^{\otimes m}\ket{1}^{\otimes m}
        =
        \frac{1}{\sqrt{2^m}}\sum_{z\in\{0,1\}^m}(-1)^{|z|}\ket{z}.
    \end{equation}
    Applying the final Hadamards to the ancillas therefore gives
    \begin{equation*}
        \ket{\Phi_3}
        =
        \frac{1}{\sqrt{2^{m+1}}}
        \sum_{z\in\{0,1\}^m}
        \left(
            1+(-1)^{|z|}(-\ri)e^{\ri m(\theta-\vtheta)}
        \right)
        \ket{z}\ket{\phi}^{\otimes m}.
    \end{equation*}
    Thus the probability of observing a particular ancilla string $z$ is
    \begin{align*}
        p(z)
        &=
        \left\|
            \frac{1}{\sqrt{2^{m+1}}}
            \left(
                1+(-1)^{|z|}(-\ri)e^{\ri m(\theta-\vtheta)}
            \right)
            \ket{\phi}^{\otimes m}
        \right\|^2 \\
        &=
        \frac{1}{2^{m+1}}
        \left|
            1+(-1)^{|z|}(-\ri)e^{\ri m(\theta-\vtheta)}
        \right|^2 \\
        &=
        \frac{1}{2^{m}}
        \left(
            1+(-1)^{|z|}\sin(m(\theta-\vtheta))
        \right).
    \end{align*}
    For the parity random variable $Y=(-1)^{|z|}$,
    \begin{align*}
        \bbE[Y]
        &=
        \sum_{z\in\{0,1\}^m}(-1)^{|z|}p(z) \\
        &=
        \frac{1}{2^{m}}\sum_{z}
        \left[
            (-1)^{|z|}
            +
            \left((-1)^{|z|}\right)^2
            \sin(m(\theta-\vtheta))
        \right]
        =
        \sin(m(\theta-\vtheta)),
    \end{align*}
    because exactly half of all bit strings have even parity and half
    have odd parity.
\end{proof}

For $m=1$, the ancilla state reduces to $\ket{+}=(\ket{0}+\ket{1})/\sqrt{2}$, and the parity outcome is simply the single-qubit measurement outcome $X=(-1)^b$. The circuit therefore reduces to the standard Hadamard test, giving $\bbE[X]=\sin(\theta-\vtheta)$.


\section{Shot-count bounds for standard and entanglement-assisted Hadamard tests}
\label{app:shot_count_bounds}

We first establish a general inverse-sine shot-count bound, then apply it to prove the bounds for the standard and entanglement-assisted Hadamard tests in Lemmas~\ref{lem:standard_hadamard_shot_count} and~\ref{lem:parallel_hadamard_shot_count}, respectively.

\subsection{General inverse-sine shot-count bound}
\label{app:optimized_inversion_radius}

We now record the inverse-sine shot count bound used throughout the paper. To motivate the choice below, first consider the unamplified case $m=1$, with signal mean $\sin\theta$ and $|\theta|\le\Delta$. We seek a tolerance $r$ such that, whenever the sine estimate deviates by at most $r$, the resulting phase error is at most $\epsilon$.
On the interval $[-s,s]$, the function $\arcsin$ is $(1-s^2)^{-1/2}$-Lipschitz. If $|\hS-\sin\theta|\le r$, then
\begin{equation*}
    |\hS|
    \le
    \sin\Delta+r.
\end{equation*}
Thus, requiring $r\le s-\sin\Delta$ keeps $\hS$ inside $[-s,s]$. If, in addition, $r\le\epsilon\sqrt{1-s^2}$, then
\begin{equation*}
    |\htheta-\theta|
    =
    |\arcsin(\hS)-\arcsin(\sin\theta)|
    \le
    \frac{|\hS-\sin\theta|}{\sqrt{1-s^2}}
    \le
    \epsilon.
\end{equation*}
Hence, for fixed $s$, the largest admissible tolerance is
\begin{equation*}
    r(s)
    =
    \min\left\{s-\sin\Delta,\ \epsilon\sqrt{1-s^2}\right\}.
\end{equation*}
The first term keeps the estimator on the branch where $\arcsin$ is well defined, while the second term guarantees the desired phase accuracy. Optimizing over $s$ balances these two constraints, and the resulting optimal tolerance is exactly the value $r$ defined in the following lemma, after substituting $\Delta\to m\Delta$ and $\epsilon\to q/m$.

\begin{lemma}[Inverse-sine shot count]
    \label{lem:app_optimized_inverse_sine_shot_count}
    Let $m \ge 1$ be an integer amplification factor, and let $\Delta
    \ge 0$ be a local error bound. Fix a branch margin
    $\chi\in(0,\pi/2)$ and assume $m\Delta\le\chi$. Let $X \in \{\pm1\}$
    be a random variable satisfying
    \begin{equation*}
        \bbE[X] = \mu,\qquad |\mu|\le\sin(m\Delta).
    \end{equation*}
    Let $\hS=N^{-1}\sum_{j=1}^N X^{(j)}$ be the empirical mean of $N$
    shots. For a phase-accuracy target $\epsilon>0$, set $q:=m\epsilon$
    and define
    \begin{equation*}
        r := \frac{q\left(\sqrt{\cos^2(m\Delta)+q^2}-q\sin(m\Delta)\right)}{1+q^2}.
    \end{equation*}
    If the number of samples satisfies
    \begin{equation*}
        N\ge\frac{2}{r^2}\ln\frac{2}{p_{\fail}},
    \end{equation*}
    then
    \begin{equation*}
        \bbP\!\left(\frac{1}{m}\left|\arcsin(\hS)-\arcsin(\mu)\right|>\epsilon\right) \le p_{\fail}.
    \end{equation*}
    As $q\to0$,
    \begin{equation*}
        r=q\cos(m\Delta) - q^2\sin(m\Delta) + \bigO(q^3).
    \end{equation*}
\end{lemma}

\begin{proof}
    Write
    \begin{equation*}
        a:=\sin(m\Delta),\qquad
        c:=\cos(m\Delta)>0 .
    \end{equation*}
    Let $s\in(a,1)$ be the unique solution of
    \begin{equation*}
        s-a=q\sqrt{1-s^2}.
    \end{equation*}
    Solving this quadratic equation gives
    \begin{equation*}
        s = \frac{a+q\sqrt{c^2+q^2}}{1+q^2},
    \end{equation*}
    and direct substitution shows that the value
    \begin{equation*}
        r:=s-a=q\sqrt{1-s^2}
    \end{equation*}
    equals the expression stated in the lemma. On the event
    $|\hS-\mu|\le r$, we have
    \begin{equation*}
        |\hS| \le |\mu|+|\hS-\mu| \le a+s-a = s,
    \end{equation*}
    so $\mu,\hS\in[-s,s]$. Since $\arcsin(x)$ has derivative
    $(1-x^2)^{-1/2}$, it is $(1-s^2)^{-1/2}$-Lipschitz on the interval $[-s, s]$.
    Hence
    \begin{equation*}
        \frac{1}{m}
        \left|\arcsin(\hS)-\arcsin(\mu)\right|
        \le
        \frac{|\hS-\mu|}{m\sqrt{1-s^2}}
        \le
        \frac{r}{m\sqrt{1-s^2}}
        =
        \frac{q}{m}
        =
        \epsilon .
    \end{equation*}
    Therefore
    \begin{equation*}
        \left\{\frac{1}{m} \left|\arcsin(\hS)-\arcsin(\mu)\right| > \epsilon \right\} \subseteq \{|\hS-\mu|>r\}.
    \end{equation*}
    Hoeffding's inequality~\cite{hoeffding1963probability} for averages of $\{\pm1\}$-valued random
    variables gives
    \begin{equation*}
        \bbP\!\left(|\hS-\mu|>r\right) \le 2\exp\!\left(-\frac{N r^2}{2}\right) \le p_{\fail}.
    \end{equation*}

    Finally, since $m\Delta\le\chi<\pi/2$, we have $c\ge\cos\chi>0$.
    Hence the Taylor expansions
    \begin{equation*}
        \sqrt{c^2+q^2}=c+\frac{q^2}{2c}+\bigO (q^4), \qquad
        \frac{1}{1+q^2}=1+\bigO(q^2).
    \end{equation*}
    Substituting these expansions into the expression for $r$ gives
    \begin{equation*}
        r = q c-q^2 a+\bigO(q^3) = q\cos(m\Delta)-q^2\sin(m\Delta) +\bigO(q^3).
    \end{equation*}
\end{proof}

\subsection{Standard Hadamard test}\label{app:proof_of_std_hadamard}

\noindent \textbf{Lemma~~\ref{lem:standard_hadamard_shot_count} (Shot count for the SHT)} Let $\ket{\phi}$ be an eigenstate of a unitary $U$ with $U\ket{\phi}=e^{\ri\theta}\ket{\phi}$, and let $\vtheta$ be a reference phase satisfying $|\theta-\vtheta|\le\Delta$ for some constant $\Delta<\pi/2$. To estimate $\theta$ via the standard Hadamard test within additive error $\epsilon$ and with success probability at least $1-p_{\fail}$, it suffices to take a shot count
\begin{equation*}
    N_{\std}
    =
    \left\lceil\frac{2}{r_{\std}^2}\,\ln\frac{2}{p_{\fail}}\right\rceil,
\end{equation*}
where $r_{\rm std} = \epsilon\left(\sqrt{\cos^2\Delta+\epsilon^2} -\epsilon\sin\Delta\right)/(1+\epsilon^2)$.

\begin{proof}
    Apply Lemma \ref{lem:app_optimized_inverse_sine_shot_count} with $m=1$
    and $\mu=\sin(\theta-\vtheta)$. Since $\Delta<\pi/2$, the local
    branch satisfies $\arcsin(\mu)=\theta-\vtheta$, so the inverse-sine
    error in the lemma is exactly $|\htheta-\theta|$.
\end{proof}
\subsection{Entanglement-assisted Hadamard test}\label{app:proof_of_parallel_hadamard_shot_count}

\noindent \textbf{Lemma~\ref{lem:parallel_hadamard_shot_count} (Shot count for the EHT)}
Let $\ket{\phi}$ be an eigenstate of a unitary $U$ with $U\ket{\phi}=e^{\ri\theta}\ket{\phi}$, and let $\vtheta$ be a reference phase satisfying $|\theta-\vtheta|\le \Delta$ for some constant $\Delta<\frac{\pi}{2m}$. Fix a branch margin $\chi<\pi/2$ and choose an integer amplification $m$ such that $m\Delta\le\chi$. To estimate the phase $\theta$ via the $m$-qubit GHZ state EHT within additive error $\epsilon$ and with success probability at least $1-p_{\fail}$, it suffices to take a shot count
\begin{equation*}
    N_{\rm EHT}=\left\lceil\frac{2}{r_{\rm EHT}^2}\ln\frac{2}{p_{\fail}}\right\rceil,
\end{equation*}
where 
\begin{align*}
    r_{\rm EHT} = \frac{m\epsilon\left(\sqrt{\cos^2(m\Delta)+m^2\epsilon^2}-m\epsilon\sin(m\Delta)\right)}{1+m^2\epsilon^2}.
\end{align*}

\begin{proof}
    Apply Lemma \ref{lem:app_optimized_inverse_sine_shot_count} with $\mu=\sin(m(\theta-\vtheta))$. Since $|m(\theta-\vtheta)|\le\chi<\pi/2$, the local branch satisfies $\arcsin(\mu)=m(\theta-\vtheta)$, so the inverse-sine error in the
    lemma gives the stated phase-error bound after division by $m$.
\end{proof}

\section{Upper bound of preparation error}
\label{app:upper_bound_of_preparation_error}

\begin{lemma}
    \label{lem:target_rotated_overlap_bounds}
    Assume the infidelity $\eta<1/2$. Define
    \begin{equation*}
        \xi
        :=
        e^{-\ri\theta_{\tar}}\bra{\psi}U\ket{\psi}.
    \end{equation*}
    Then
    \begin{equation*}
        |\xi| \ge 1-2\eta,
        \qquad
        |\arg(\xi)| \le \arcsin\!\left(\frac{\eta}{1-\eta}\right),
    \end{equation*}
    where the argument is taken on the local branch containing the positive
    real axis. 
\end{lemma}

\begin{proof}
    Since $\ket{\phi_{\tar}^{\perp}}$ is orthogonal to the target
    eigenstate and $U\ket{\phi_{\tar}}=e^{\ri\theta_{\tar}}
    \ket{\phi_{\tar}}$, the cross terms vanish and
    \begin{align*}
        \xi
        =
        e^{-\ri\theta_{\tar}}
        \bra{\psi}U\ket{\psi}
        =
        (1-\eta)
        +
        \eta\, 
        e^{-\ri\theta_{\tar}}
        \bra{\phi_{\tar}^{\perp}}U
        \ket{\phi_{\tar}^{\perp}} .
    \end{align*}
    Thus $\xi=(1-\eta)+\eta\zeta$ with
    \begin{equation*}
        \zeta
        :=
        e^{-\ri\theta_{\tar}}
        \bra{\phi_{\tar}^{\perp}}U
        \ket{\phi_{\tar}^{\perp}},
        \qquad
        |\zeta|\le1,
    \end{equation*}
    because $U$ is unitary and
    $\ket{\phi_{\tar}^{\perp}}$ is normalized. Therefore $\xi$ lies
    in the closed disk of radius $\eta$ centered at $1-\eta$ on the real
    axis. The reverse triangle inequality gives
    \begin{equation*}
        |\xi|
        \ge
        (1-\eta)-\eta
        =
        1-2\eta.
    \end{equation*}
    Since $\eta<1/2$, this disk does not contain the origin. The largest
    possible absolute argument is achieved by a ray from the origin
    tangent to the disk. In the resulting right triangle, the opposite
    side has length $\eta$ and the hypotenuse from the origin to the disk
    center has length $1-\eta$. Hence
    \begin{equation*}
        |\sin(\arg \xi)|
        \le
        \frac{\eta}{1-\eta}.
    \end{equation*}
    The disk lies in the right half-plane, so
    $|\arg \xi|<\pi/2$, and the claimed argument bound follows by
    applying $\arcsin$.
\end{proof}
\section{Expectation value of entanglement-assisted Hadamard test with imperfect state preparation}\label{app:proof_of_parallel_hadamard_shot_count_imperfect}

\begin{lemma}[Measured parity for imperfect states]
\label{lem:parity}
For the parallel $m$-copy circuit, the measured ancilla parity
$Y=(-1)^{|z|}$ satisfies
\begin{equation}
  \mathbb{E}[Y]=\rho^{m}\sin\!\big(m(\theta_\psi-\vartheta)\big).
\end{equation}
\end{lemma}

\begin{proof}
    For the parallel circuit, the state immediately after the controlled
    applications of $U$ and the reference block is
    \begin{equation*}
        \frac{1}{\sqrt{2}}
        \left(
            \ket{0}^{\otimes m}\ket{\psi}^{\otimes m}
            +
            (-\ri)e^{-\ri m\vtheta}
            \ket{1}^{\otimes m}(U\ket{\psi})^{\otimes m}
        \right).
    \end{equation*}
    Applying $H^{\otimes m}$ to the ancillas and using
    \eqref{eq:multi_hadamard_relation}, the probability of observing
    $z\in\{0,1\}^m$ is
    \begin{align*}
        p(z)
        &=
        \frac{1}{2^m}
        \left(
            1
            +
            (-1)^{|z|}
            \mathrm{Re}\!\left(
                (-\ri)e^{-\ri m\vtheta}
                (\bra{\psi}U\ket{\psi})^m
            \right)
        \right) \\
        &=
        \frac{1}{2^m}
        \left(
            1+(-1)^{|z|}\rho^m
            \sin\bigl(m(\theta_{\psi}-\vtheta)\bigr)
        \right).
    \end{align*}
    Averaging $Y=(-1)^{|z|}$ over this distribution gives the stated
    expectation. 
\end{proof}
\section{Amplitude-induced bias error bound}\label{app:proof_of_amplitude_error_bound}

\noindent \textbf{Theorem~\ref{thm:single_round_imperfect} (Amplitude bias error bound)}
Let $U$ be a unitary and let $\ket{\psi}$ be a prepared state with $\langle\psi|U|\psi\rangle=\rho e^{i\theta_\psi}$, where $\theta_\psi$ is the effective phase and $\rho=\lvert\langle\psi|U|\psi\rangle\rvert$ is the signal amplitude. Assume that $\rho$ is lower bounded by a known constant $\rho_0\in[0,1]$, i.e., $\rho\ge\rho_0$. Let $\vtheta$ be a reference phase satisfying $|\theta_\psi-\vtheta|\le\Delta$ for a reference-error bound $\Delta>0$, and choose an integer amplification $m$ and a branch margin $\chi\in(0,\pi/2)$ such that $m\Delta\le\chi$. Then the amplitude-induced bias $\varepsilon_{\mathrm{bias}}$ satisfies
\begin{equation}
    |\varepsilon_{\mathrm{bias}}|
    \le
    \frac{1-\rho_0^{m}}{m}\tan(m\Delta).
\end{equation}

\begin{proof}
    Let $\beta:=m(\theta_\psi-\vtheta)$. The branch condition gives
    $|\beta|\le m\Delta\le\chi<\pi/2$, so
    $\arcsin(\sin\beta)=\beta$.
    By the definition of the infinite-shot limit,
    \begin{equation*}
        \theta_{\psi,\infty}
        =\vtheta+\frac{1}{m}\arcsin(\rho^m\sin\beta).
    \end{equation*}
    Therefore, the amplitude-induced bias is
    \begin{align*}
        \varepsilon_\bias
        =|\theta_{\psi,\infty}-\theta_\psi|
        =\frac{1}{m}
        \left|\arcsin(\rho^m\sin\beta) - \arcsin(\sin\beta)\right|.
    \end{align*}
    Note that
    \begin{equation*}
        \arcsin(\rho^m\sin\beta) - \arcsin(\sin\beta) =    \int_{\sin\beta}^{\rho^m\sin\beta} \frac{\mathrm{d}u}{\sqrt{1-u^2}}.
    \end{equation*}
    The integration interval has
    length $(1-\rho^m)|\sin\beta|$, and every point in it has absolute
    value at most $\sin(m\Delta)$. Therefore
    \begin{equation*}
        |\veps_{\bias}|
        \le
        \frac{(1-\rho^m)|\sin\beta|}
        {m\sqrt{1-\sin^2(m\Delta)}}
        \le
        \frac{(1-\rho^m)\sin(m\Delta)}
        {m\cos(m\Delta)}
        =
        \frac{1-\rho^m}{m}\tan(m\Delta).
    \end{equation*}
    Since $\rho\ge\rho_0$, we have $1-\rho^m\le 1-\rho_0^m$, and the
    claimed bound
    \begin{equation*}
        |\veps_{\bias}| \le \frac{1-\rho_0^m}{m}\tan(m\Delta)
    \end{equation*}
    follows.
\end{proof}


\section{Amplification selection}
\label{app:Optimal_amplification_with_imperfect_state_preparation}

We state an amplification-selection rule that minimizes the certified device-restart count in each round for both exact and imperfect eigenstate preparation. Let $\Delta>0$ be the current reference-error bound, and let $\epsilon_{\stat}>0$ and $\epsilon_{\bias}\ge0$ be the statistical and bias budgets. For an amplitude lower bound $\rho_0\in[0,1]$, recall the amplitude-bias bound
\begin{equation*}
    B(m;\Delta,\rho_0)
    :=
    \frac{1-\rho_0^m}{m}\tan(m\Delta)
\end{equation*}
and define the feasible set
\begin{equation*}
    \mathcal M_{\mathrm{feas}}
    :=
    \left\{
        m\in\{1,\ldots,m_{\hw}\}:
        m\Delta\le\chi,\quad
        B(m;\Delta,\rho_0)\le\epsilon_{\bias}
    \right\}.
\end{equation*}
For exact eigenstate preparation, set $\rho_0=1$, so that $B(m;\Delta,1)=0$ for every admissible amplification. The entire round-wise accuracy budget can then be assigned to statistical error: $\epsilon_{\bias}=0$ and $\epsilon_{\stat}=\epsilon_t$ (equivalently, $c_{\bias}=0$ and $c_{\stat}=\gamma$ when $\epsilon_t=\gamma\Delta$). The feasible set reduces to
\begin{equation*}
    \mathcal M_{\mathrm{feas}}
    =\{1,\ldots,m_{\mathrm{search}}\},
    \qquad
    m_{\mathrm{search}}
    :=\min\{m_{\hw},\lfloor\chi/\Delta\rfloor\},
\end{equation*}
provided that $m_{\mathrm{search}}\ge1$. Thus only the branch and hardware constraints remain. For imperfect preparation, the bias constraint additionally restricts this set.

For each $m\in\mathcal M_{\mathrm{feas}}$, let $r_m$ be the inverse-sine value from Appendix~\ref{app:optimized_inversion_radius}, evaluated at statistical accuracy $\epsilon_{\stat}$. For a per-round failure probability $p_{\mathrm{round}}\in(0,1)$, define the certified shot and restart counts
\begin{equation*}
    N_m
    :=
    \left\lceil
        \frac{2}{r_m^2}\ln\frac{2}{p_{\mathrm{round}}}
    \right\rceil,
    \qquad
    R_m
    :=
    \left\lceil
        \frac{N_m}{\kappa(m)}
    \right\rceil,
\end{equation*}
where $\kappa(m)$ is the packing capacity. We choose the feasible amplification with the smallest packed restart count $R_m$. If several amplifications have the same restart count, we choose the one with smaller amplification. Without the integer ceilings, this restart objective is equivalent to maximizing $\kappa(m)r_m^2$ rather than $r_m$ alone.

Within the range allowed by the branch condition and hardware cap, the bias bound $B(m;\Delta,\rho_0)$ is nondecreasing in $m$ for $m\Delta<\pi/2$. Consequently, if an amplification $m$ satisfies the bias constraint, then every smaller positive amplification also satisfies it. The feasible candidates therefore have the form $\{1,\ldots,m_{\mathrm{feas}}\}$, and the largest bias-feasible value $m_{\mathrm{feas}}$ can be found by binary search. The restart cost itself need not be monotone because of inverse-sine conditioning and integer packing, so all these feasible candidates must still be compared. The resulting selection rule is summarized in Algorithm~\ref{alg:min_restart_amplification}.

\begin{algorithm}[htbp]
\renewcommand{\algorithmicrequire}{\textbf{Input:}}
\renewcommand{\algorithmicensure}{\textbf{Output:}}
    \caption{\textsc{MinRestartAmp}: restart-minimizing amplification selection}
    \label{alg:min_restart_amplification}
    \begin{algorithmic}[1]
        \REQUIRE Reference-error bound $\Delta$, statistical budget $\epsilon_{\stat}$, bias budget $\epsilon_{\bias}$, failure probability $p_{\fail}$, amplitude lower bound $\rho_0$ and branch margin $\chi$, hardware constraint $m_{\hw}$ and packing capacities $\kappa(m)$
        \STATE $m_{\mathrm{search}}\gets
            \min\{m_{\hw},\lfloor\chi/\Delta\rfloor\}$
        \IF{$m_{\mathrm{search}}<1$ \textbf{or}
            $B(1;\Delta,\rho_0)>\epsilon_{\bias}$}
            \STATE \textbf{Return infeasible}
        \ENDIF
        \IF{$\rho_0=1$}
            \STATE $m_{\mathrm{feas}}\gets m_{\mathrm{search}}$
        \ELSE
            \STATE Use binary search over $m=1,\ldots,m_{\mathrm{search}}$ to find the largest bias-feasible amplification $m_{\mathrm{feas}}$ such that $B(m;\Delta,\rho_0)\le\epsilon_{\bias}$
        \ENDIF
        \STATE Choose $m_{\mathrm{best}}$ in $m=1,\ldots,m_{\mathrm{feas}}$ requiring the fewest packed restarts $R_m$ and smaller amplification
        \STATE Return $m_{\mathrm{best}}$
    \end{algorithmic}
\end{algorithm}

\end{document}